\documentclass[10pt,a4paper]{article}
\usepackage[utf8]{inputenc}
\usepackage[T1]{fontenc}
\usepackage[british]{babel}

\usepackage{amsmath}
\usepackage{amsthm}
\usepackage{amsfonts}
\usepackage{amssymb}
\usepackage{mathtools}
\usepackage{stmaryrd}
\usepackage{physics}
\usepackage{graphicx}
\usepackage{bbold}
\usepackage{enumitem}
\usepackage [autostyle, english = british]{csquotes}
\MakeOuterQuote{"}
\usepackage{float}
\usepackage[toc,page,title]{appendix}
\usepackage{multirow}
\usepackage{array}

\usepackage[font=small,labelfont=bf]{caption}
\usepackage{subcaption}
\usepackage{sidecap}
\usepackage{eucal}
\usepackage{changepage}
\usepackage{setspace}
\usepackage{placeins} 
\usepackage{titling}
\usepackage{authblk} 
\usepackage[super]{nth}
\usepackage{xcolor} 
\usepackage{tikz}
\usetikzlibrary{backgrounds}
\allowdisplaybreaks

\usepackage[backend=biber, style=alphabetic, sorting=nyt, date=year, doi=true, isbn=false, url=false, maxbibnames=6]{biblatex}
\AtBeginBibliography{\small}

\usepackage[colorlinks=true,linkcolor=blue,citecolor=magenta]{hyperref}
\usepackage[nameinlink,capitalise]{cleveref}
\usepackage{aliascnt}

\newtheorem{theorem}{Theorem}
\crefname{theorem}{Theorem}{Theorems}
\newaliascnt{lemma}{theorem}
\newtheorem{lemma}[lemma]{Lemma}
\aliascntresetthe{lemma}
\crefname{lemma}{Lemma}{Lemmas}
\newaliascnt{corollary}{theorem}
\newtheorem{corollary}[corollary]{Corollary}
\aliascntresetthe{corollary}
\crefname{corollary}{Corollary}{Corollaries}
\newaliascnt{definition}{theorem}

\aliascntresetthe{definition}
\crefname{definition}{Definition}{Definitions}
\newaliascnt{conjecture}{theorem}

\aliascntresetthe{conjecture}
\crefname{conjecture}{Conjecture}{Conjectures}
\newaliascnt{remark}{theorem}

\aliascntresetthe{remark}
\crefname{remark}{Remark}{Remarks}
\newaliascnt{proposition}{theorem}

\aliascntresetthe{proposition}
\crefname{proposition}{Proposition}{Propositions}
\newaliascnt{condition}{theorem}

\aliascntresetthe{condition}
\crefname{condition}{Condition}{Conditions}
\newaliascnt{example}{theorem}
\newtheorem{example}[example]{Example}
\aliascntresetthe{example}
\crefname{example}{Example}{Examples}
\crefname{equation}{eq.}{eqs.}

\numberwithin{equation}{section}

\crefformat{footnote}{#2\footnotemark[#1]#3}

\renewcommand{\i}{\mathrm{i}}
\newcommand{\mc}[1]{\mathcal{#1}}
\renewcommand{\d}{\mathrm{d}}

\definecolor{myred}{RGB}{221,0,0}
\definecolor{myblue}{RGB}{0,0,221}
\definecolor{citegreen}{RGB}{0,165,0}

\makeatletter
\renewcommand\@fnsymbol[1]{\ifcase#1\or *\or §\or ¶\or d\else *\fi}
\makeatother

\title{A sharper Magnus expansion bound woven in binary branches}

\author[1]{Harriet Apel}
\author[1]{Toby Cubitt}
\author[2]{Emilio Onorati}
\affil[1]{Department of Computer Science, University College London, UK}
\affil[2]{Zentrum Mathematik, Technische Universit{\"a}t M{\"u}nchen, DE}

\hypersetup{
	pdftitle={A sharper Magnus expansion bound woven in binary branches},
	pdfauthor={Harriet Apel, Toby Cubitt, Emilio Onorati}
}

\date{}

\begin{document}

\maketitle
\vspace{-0.75cm}

\begin{abstract}
	The Magnus expansion provides an exponential representation of one-parameter operator families, expressed as a series expansion in their generators.
	This is particularly relevant in quantum mechanics for describing a unitary evolution determined by a time-dependent Hamiltonian generator of the dynamics.
	The solution is constructed as a series expansion in terms of increasingly complex nested commutators that rapidly become challenging to compute directly.
	This work establishes a universal upper bound scaling as $\mc O ((A(t)/\xi)^n n^{-3/2})$ on the error incurred when the Magnus expansion is truncated at an arbitrary order.
	The factor $n^{-3/2}$ sharpens the exponential suppression governed by the integrated norm of the generator $A(t)$ and the convergence radius $\xi$, and we prove that this dependence on the order $n$ is optimal for the corresponding majorant.
	The main technical ingredient of the proof is the binary tree representation introduced by Iserles and N{\o}rsett from which we derive a recursion formula to delimit the magnitude of any term in the expansion.
	We complement our analytic bound with explicit computation of the first tree coefficients.
	With these findings we aim to contribute to the understanding of the accuracy and limitations of the Magnus expansion technique, and to provide a sharper bound for approximating quantum dynamics when information on the generator is not available.
	\end{abstract}

\section{Introduction}

The Magnus expansion is an exponential representation of the solution of an ordinary linear differential equation~\cite{Magnus1954}.
Magnus himself emphasised that his work was inspired by Friedrichs’ operator theory in quantum mechanics~\cite{Friedrichs1953}, and cited Feynman’s operator calculus~\cite{Feynman1951} as antecedent.
Thus, from its origin, the Magnus expansion has been closely tied to quantum theory.
In particular, it provides a natural solution for the Schr\"{o}dinger equation, where it expresses the exponent of the time-evolution operator of a quantum system described by a time-dependent Hamiltonian operator.
That is, the Magnus expansion serves as a continuous-time analogue of the Baker–Campbell–Hausdorff formula~\cite{rossmann2002lie}, where the dynamics is not restricted to piecewise constant.

The Magnus exponent can be expanded in a power series of nested commutators of the Hamiltonian at different times.
Initially the series terms were presented recursively~\cite{Burum}, while later on \cite{Iserles} provided an explicit form of the expansion. 
An understanding of the limits of convergence of the series has been elucidated~\cite{moan1998, moan2002backward, Moan2008, Sanchez2011}.
In practice, the expansion can be truncated at finite order within an explicitly established convergence radius.

A key advantage of the Magnus expansion, compared to other similar expansions such as the Dyson series~\cite{DysonOriginal}, is that it preserves the unitarity of the time-evolution operator at any truncation order.\footnote{This can be seen from the form of the $n$-th term in the expansion that we illustrate later in \cref{eqn_Magnus_tree}, since if $\i H(t)$ belongs to a Lie algebra, then sums of integrals over nested commutators of $\i H(t)$ also stay within the Lie algebra and consequently any truncation of the Magnus expansion does too. This is important physically, since it implies that any truncation of the series will still generate meaningful unitary dynamics.}
Importantly, it captures the higher-order terms that arise from the commutators of the Hamiltonian at different times, which are instead not explicitly expressed in the Trotter approximation~\cite{Trotter1959}.
Hence, the Magnus expansion is an important tool for several applications~\cite{blanes09} like high precision simulations~\cite{LeimkuhlerReich2005,Sharma,chen2023,Bosse2025} (particularly for Hamiltonians exhibiting strong non-commutativity between different times), studies of dissipative dynamics~\cite{Bezjav,Reimer}, matter--radiation approximation~\cite{Macri2023,ture2024,Dey25}, quantum optimal control~\cite{Ma2020,marecek2020}, and recent work in quantum field theory~\cite{brandhuber2026}.
The complexity of the Hamiltonian, the order of the truncation, and the time interval being considered all affect the quality of the Magnus expansion approximation, and have been investigated in different works~\cite{Bialynicki,Hochbruck2003,Sanchez2011,Mori16,Kuwahara,onebound}.
Recent work has extended the Floquet--Magnus expansion to the setting of unbounded Hamiltonians~\cite{Burgarth2026}.

In this note, we establish an explicit universal upper bound on the truncation error at arbitrary order, presented here in closed form, without assumptions on the generator beyond norm integrability, which need not even be a Hamiltonian, and with a new $n^{-3/2}$ polynomial suppression factor sharpening the exponential decay governed by the Moan--Blanes convergence radius $\xi$.
We further show that this scaling is optimal for the corresponding tree majorant, i.e. the positive scalar sequence dominating the tree expansion.
The asymptotic prefactor is exactly $\sqrt{2\xi}/2$, while direct computation in the pre-asymptotic regime and eventual monotonicity provide strong numerical evidence for a global prefactor of $2$, improving on the $\pi$ prefactor in the classical bound~\cite{moan2002backward,blanes09}.
As an immediate application, the truncation bound is used in recent work on Floquet-universal simulation to control the infinite Magnus tail of a time-periodic Hamiltonian and choose the physical amplitudes to isolate the leading order contribution~\cite{OnoratiFloquet}.

A central feature of our construction is the connection between the binary-tree representation of nested commutators and singularity analysis in analytic combinatorics.
The binary trees produce a recursive structure whose generating function determines the large-order behaviour.
Its dominant square-root singularity yields the $n^{-3/2}$ prefactor that refines the exponential decay.

\medskip

In \cref{sec:Magnus}, we review the Magnus expansion as the exponential representation of the solution of a differential equation, that is particularly relevant e.g.\ in the context of a unitary evolution of a time-dependent generator in quantum dynamics.
In \cref{sec:previous_bounds} we then summarise the previous bounds to provide context for our result, before we present our new truncation bound in \cref{sec:main_results} with an outline of the proof.
The remainder of the paper gives the necessary technical background before setting out the full proof of this bound.

\subsection{The Magnus expansion}\label{sec:Magnus}

The time-dependent Schr\"{o}dinger equation describes the evolution of a state governed by Hamiltonian dynamics,
\begin{equation}\label{eq:QM}
	\i \frac{d}{dt} \psi(t) = H(t) \psi(t),
\end{equation}
and is an ordinary linear differential equation.
The solution is given by
\begin{equation}
	\psi(t) = U(t)\psi(t=0),
\end{equation}
where $U(t)$ is the time-evolution operator.
For a general time-dependent Hamiltonian, $H(t)$, the evolution operator is expressed in terms of a complex time-ordered exponential,
\begin{equation}
U(t) = \mathbf T \left\{\exp \left[-\i \int_0^t H(\tau)d\tau \right] \right\}.
\end{equation}
where $\mathbf T$ is the time-ordering operator. Throughout this work, $H:[0,T]\to\mathcal{B}(\mathcal{H})$ denotes a norm-integrable family of bounded operators.

\medskip 

The Magnus series~\cite{Magnus1954} provides an exponential representation of the solution of such equations.
Applying this approach, the evolution operator can be expressed as
\begin{equation}
U(t) = \exp\left[ \mathcal{M}(t)\right],
\end{equation}
where the exponent is given by an infinite series expansion with terms involving multiple integrals of nested commutators,
\begin{equation}\label{eq:Magnus_basic}
\mathcal{M}(t) = \sum_{n=1}^\infty M_n(t).
\end{equation}

Explicitly, the first four terms in the expansion read,
\begin{adjustwidth}{-1.2cm}{-1.2cm}
	\begin{align}
	M_1(t)  &= - \i \int_0^t H(t_1) dt_1\label{eqn mag term 1}\\
	M_2(t)  &= \frac{(-\i)^2}{2}\int_0^t dt_1 \int_0^{t_1}dt_2 \left[H(t_1),H(t_2) \right]\label{eqn mag term 2}\\
	M_3(t)  &= \frac{(-\i)^3}{6} \int_0^t dt_1 \int_0^{t_1}dt_2 \int_0^{t_2}dt_3 \left(\left[H(t_1), \left[H(t_2),H(t_3) \right] \right]
	+ \left[ H(t_3),\left[ H(t_2),H(t_1)\right]\right] \right) \label{eq:3rd_magnus_1}\\
	M_4(t)   &= \frac{(-\i)^4}{12} \int_0^t dt_1 \int_0^{t_1} dt_2 \int_0^{t_2} dt_3 \int_0^{t_3} dt_4 \left(\left[\left[\left[H(t_1),H(t_2) \right] ,H(t_3)\right],H(t_4) \right] \right. \\
	&+ \left[H(t_1),\left[ \left[ H(t_2),H(t_3)\right],H(t_4)\right] \right]
	+ \left[H(t_1),\left[H(t_2),\left[H(t_3),H(t_4)\right]\right]\right]
	+\left. \left[ H(t_2), \left[ H(t_3),\left[H(t_4),H(t_1)\right]\right]\right]\right).\label{eqn mag term 4}
	\end{align}
\end{adjustwidth}
The above gives a sense of higher-order terms growing increasingly complex, hence more systematic methods are needed to calculate the expansion to a given order.
Our approach will be the binary trees representation developed in \cite{Iserles,IserlesTree}, which we recap in \cref{sec:binary_trees_explained}.

We now consider \emph{truncations} of the Magnus expansion, that is the sum of the first $N$ (say) terms:
\begin{equation}
	\mathcal{M}^{(N)}(t) = \sum_{n=1}^N M_n (t) \ .
\end{equation}
\cite{Magnus1954} demonstrated that the Magnus expansion is equal to the first order Magnus approximation if and only if the following condition is satisfied for all times $t$,
\begin{equation}
	\left[H(t), \int_0^t H(\tau) d\tau \right] = 0.
\end{equation}
This condition is not fulfilled for general time-dependent Hamiltonians, therefore generally truncation will produce an error -- the quantity that we aim to bound universally in this work.

\medskip

\paragraph{Remark on generality of our results}
In this work we adopt the notation~$H$ for the operators in the nested commutators, tied to the conventional description for Hamiltonians in the Schr\"{o}dinger equation.
We emphasise however that our results \cref{lemma:bound_magnus_term,{Main_Thm}} apply in full generality to differential equations of the form  \cref{eq:QM} for one-parameter operator families, without requiring Hermiticity as postulated by quantum mechanics.
Thus our bound is not confined to instances of quantum mechanics and the Lie algebra of the unitary group, but extends equally to all finite-dimensional normed Lie algebras with $\norm{[X,Y]} \leq 2 \norm{X}\norm{Y}$.

\subsection{Previous bounds and convergence radius}\label{sec:previous_bounds}

In the structure-free case, \cite{moan2002backward} derived a bound on the individual terms of the series,
\begin{equation}\label{eq:reference_bound_on_Magnus}
	\norm{M_n(t)}_\mathrm{op} \leq \pi \left(\frac{1}{\xi} \int_0^t \norm{H(s)}_\mathrm{op}  \, ds \right)^n ,
\end{equation}
for a specific constant $\xi$ expressing the convergence radius of the expansion.
This quantity has been investigated in \cite{Blanes98,moan1998}, who independently obtained the same value, approximated by
\begin{equation}\label{convergence_radius}
	\xi 
	=
	\frac{1}{2}\int_0^{2\pi}\frac{\d u}{2+\frac u2-\frac u2\cot\!\left(\frac u2\right)}
	\approx 1.086869.
\end{equation}
For matrix operators, the radius has been further enlarged up to $\pi$~\cite{Moan2008,Casas2007}, whereas \cite{lakos2025} has shown that Lie-algebraic reductions can improve the convergence radius with a lower limit of $1.2135$ in the operator-norm convention used here.

Recently, a new bound on the approximation $U_N$ of a unitary evolution based on truncation of the Magnus expansion has been derived in \cite[Theorem~1]{fang2025magnus}, highlighting the timely relevance of such results:
\begin{equation}
	\norm{U-U_N}_\mathrm{op} \leq  \mathrm{const} \cdot
	\left( \max_{N+1 \leq q \leq N^2+2N}  \alpha_{\mathrm{comm},q}^{1/q} (H) \right)^{N+1} \, t^{N+1} \ ,
\end{equation}
where $\alpha_{\mathrm{comm},q}(H)$ is the norm of the largest commutator over the set of all nested commutators of depth~$q$ and times $t_1,\dots, t_q$ evaluated on inputs $H(t_1), \dots, H(t_q)$. This is based on the bound of the individual Magnus terms derived in the same work (cf.~eq.~(61)),
\begin{equation}
	\norm{M_n}_\mathrm{op} \leq \frac{\alpha_{\mathrm{comm},n}(H) \, t^n}{n^2} \ .
\end{equation}
However, unless additional structural information on the commutators is available, one can only invoke the generic inequality
\begin{equation}
	\alpha_{\mathrm{comm},n}(H)\;\leq\;2^{\,n-1}\,\max_{s \in[0,t]}\|H(s)\|^n,
\end{equation}
here, which gives a factor~$2^n$ yielding a bound that is exponentially weaker than the bound derived in this work.

\smallskip

Structured instances are not the focus of this work, but for broader interest we will cite a few results.
For example, when one considers $\ell$-local Hamiltonians with interaction strength $J$, \cite[Lemma~1]{Kuwahara} has derived a scaling for the $n$-th order Magnus term,
\begin{equation}\label{eq:kuwahara_bound}
	\norm{M_n}_\mathrm{op} \leq \mathrm{const} \cdot \frac{(2 \ell  J)^n\, n!}{(n+1)^2} \ .
\end{equation}
While this refines \cref{eq:reference_bound_on_Magnus} by replacing the generic norm $\norm{H(s)}_\mathrm{op}$ with explicit coupling parameters, the result in \cref{eq:kuwahara_bound} comes at the cost of a factorial growth~$n!$ and does not provide a convergence radius.
Recent work~\cite[Theorem 1]{Sharma} considers geometrically local Hamiltonians, $H=\sum_{i=1}^m h_i$, and improves the error at truncation index~$N$ and with respect to the number of interaction terms~$m$, reducing a generic $\mc O (m^{N+1})$~dependence to a linear scaling~$\mc O(m)$.

There are also examples of error bounds for particular degrees of truncation. \cite[Theorem 3]{Wilcox1972} gives an error bound when truncating the Magnus expansion at first order.
Here the remainder of the truncation is expressed as a function of several, unevaluated nested integrals, rendering the result mainly useful for analysing differential equations of specific form.

\subsection{Main results}\label{sec:main_results}

Our bound on arbitrary terms of the Magnus expansion improves previous results in the literature presented in \cref{sec:previous_bounds} for the structure-free case.

\begin{lemma}[Upper bound for Magnus terms]\label{lemma:bound_magnus_term}
	Let $H$ be a general time-dependent Hamiltonian and define
	\begin{equation}\label{eq:A_t_definition}
		A(t) \coloneqq \int_0^t \norm{H(s)}_\mathrm{op} \, \d s .
	\end{equation}
	Then there exists a global constant $K_{\mathcal{M}}$ such that the $n$-th term in the corresponding Magnus expansion $\mathcal{M}(t) = \sum_{n=1}^\infty M_n(t)$ is upper bounded by
	\begin{equation}\label{eq:Magnus_main_bound}
		\norm{M_n(t)}_\mathrm{op}
		\leq
		K_{\mathcal{M}} \, \frac{\left(A(t)/\xi\right)^n}{n^{3/2}}
	\end{equation}
	for all $n \geq 1$, where $\xi \approx 1.086869$ is the convergence radius.
\end{lemma}
 \noindent The constant $K_{\mathcal{M}}$ is independent of $n$; the numerical evidence discussed in \cref{sec:compute_trees} supports the rounded choice 
 \begin{equation}
 	K_{\mathcal{M}}=2,
 \end{equation}
 which stays below the classical prefactor $\pi$ in \cref{eq:reference_bound_on_Magnus}. In addition, our bound gains a factor $n^{-3/2}$ in the scaling with $n$ that tightens the decay rate, while retaining the Moan--Blanes convergence radius~$\xi$ of \cref{convergence_radius}.

\medskip 

Using the above result, we establish a closed-form universal bound on the truncation error of the Magnus expansion, valid at any order:

\begin{theorem}[Error bound for truncated Magnus series]\label{Main_Thm}
	Consider a general time-dependent Hamiltonian $H$, and let $A(t)$ be defined as in \cref{eq:A_t_definition}. 
	Denote its Magnus series expansion by $\mathcal{M}(t) = \sum_{n=1}^\infty M_n(t)$ and the corresponding truncation at term $N$ by $\mathcal{M}^{(N)}(t) = \sum_{n=1}^N M_n(t)$.
	Assume that
	\begin{equation}\label{eq:convergence_condition_main}
		A(t)<\xi .
	\end{equation}
	Then the truncation error is bounded by
	\begin{equation}\label{eq:main_truncation_bound_closed}
		\norm{\mathcal{M}(t) - \mathcal{M}^{(N)}(t)}_\mathrm{op}
		\leq
		\frac{K_{\mathcal{M}}}{(N+1)^{3/2}} \frac{\left(A(t)/\xi\right)^{N+1}}{1-A(t)/\xi} ,
	\end{equation}
	where $K_{\mathcal{M}}$ is the global constant from \cref{lemma:bound_magnus_term}.
\end{theorem}
\noindent To our knowledge, our truncation bound is the tightest available without assumptions or structural information about the time-dependent Hamiltonian.
This improvement is of particular interest for truncations at small~$N$, when the compound effect of the exponential scaling is not yet manifest, which is often the relevant case in practical applications.
Additionally, in \cref{cor:stronger_bound} we provide a sharper bound for $\norm{\mathcal{M}(t) - \mathcal{M}^{(N)}(t)}_\mathrm{op}$ retaining the unsimplified infinite sum.

\emph{Proof outline:}
After recapitulating essential background on the relation between the Magnus expansion and binary trees (cfr.\ \cref{sec:binary_trees_explained}), we observe in \cref{lemma:Integral_recursion} that the direct evaluation of the nested integrals associated with the nested commutators in the Magnus expansion follows a recursion formula reflecting the structure of binary trees.
This connects to the binary tree formalism developed by Iserles and N{\o}rsett, and reproduces the same commutator weighting pattern derived in those earlier works.
In \cref{sec:combined_recursion}, we combine these two elements into what we call the \emph{tree coefficients} and derive a second recursion for them (cfr.\ \cref{thm:complete_recursion_formula}).
To establish the scaling of these coefficients, we encode them into a generating function~$f$ (cfr.\ \cref{lem:tree_generating_function}).
The recursion relation yields a differential equation for~$f$, whose local behaviour near its first singularity~$\rho$, defined by $f(\rho)=2\pi$, can be analysed explicitly.
The resulting square-root singular expansion of~$f$ allows us to apply the Flajolet--Odlyzko transfer theorem~\cite{Flajolet2013}, yielding the asymptotically tight scaling $\Theta(n^{-3/2}2^{-n}\xi^{-n})$ for the tree coefficients, as detailed in \cref{thm:asymptotic_tree_coeffs}.
In the asymptotic regime, \cref{cor:asymptotic_tree_coeffs} also identifies the exact prefactor $\sqrt{2\xi}/2$.
The tree coefficients enter the majorant in \cref{lemma:integral_norm_tree_bound}, which controls the norms of arbitrary Magnus terms.
In \cref{sect:bound}, we insert the all-order expression into the structure-free bound for the Magnus term at order~$n$, obtaining the per-term estimate in \cref{lemma:bound_magnus_term} and the resulting truncation bound stated in \cref{Main_Thm}.
Finally, in \cref{sec:compute_trees}, we compute the first tree coefficients explicitly and compare them with the scaling law obtained in the previous section, providing numerical evidence for the rounded global constant $K_{\mathcal{M}}=2$.

\medskip 

For Hamiltonian dynamics, the bound on the Magnus exponent directly controls the approximation error of the corresponding unitary propagator.

\begin{corollary}[Error bound for the truncated Magnus propagator]\label{cor:propagator_bound}
	Consider the assumptions of \cref{Main_Thm} and define the $N$-th order Magnus propagator by
	\begin{equation}\label{eq:truncated_Magnus_propagator}
		U_N(t)\coloneqq\exp\!\left[\mathcal{M}^{(N)}(t)\right].
	\end{equation}
	Then $U_N(t)$ is unitary and its distance from the exact propagator $U(t)=\exp[\mathcal{M}(t)]$ satisfies
	\begin{equation}\label{eq:propagator_truncation_bound}
		\norm{U(t)-U_N(t)}_\mathrm{op}
		\leq
		\norm{\mathcal{M}(t)-\mathcal{M}^{(N)}(t)}_\mathrm{op}
		\leq
		\frac{K_{\mathcal{M}}}{(N+1)^{3/2}}
		\frac{\left(A(t)/\xi\right)^{N+1}}{1-A(t)/\xi}.
	\end{equation}
\end{corollary}

Our results also apply when nothing about the generator is known apart from its supremum over the time interval, $h_{\max}(t)\coloneqq\sup_{s\in[0,t]}\norm{H(s)}_\mathrm{op}$, by replacing $A(t)$ with $t\,h_{\max}(t) < \xi$ in all the bounds above.

\medskip

\section{Graph theoretical analysis of the Magnus expansion}\label{sec:binary_trees_explained}

All Magnus terms consist of a linear combination of integrated commutators.
While the expressions in~\cref{eqn mag term 1,eqn mag term 2,eq:3rd_magnus_1,eqn mag term 4} seem manageable, higher-order Magnus terms become increasingly complex .
To handle this, \cite{Iserles} has given an expression for the Magnus expansion as a sum over full binary trees, which will be the main ingredient used here to establish the scaling of the decay of the Magnus terms, and consequently to bound the truncation of the series. The formalism has been previously used for other purposes, such as Knuth's rotation correspondence~\cite{Ebrahimi2014}.
This section will review this representation to provide the necessary context for the later proofs.

\medskip

A \emph{full binary tree} is a graphical object consisting of \emph{nodes} in a hierarchical structure, from bottom to top, starting from a single node called the \emph{root}. Each node is connected by edges to its \emph{successors} that must be exactly two or zero. A node without successors is called a \emph{leaf}. Any node except the root has exactly one parent node.

\begin{figure}[ht!]
	\begin{center}
		\begin{tikzpicture}[scale=0.40]

			\fill (0,0) circle (6pt);
			\fill (-3,2) circle (6pt);
			\fill (3,2) circle (6pt);
			\fill (-4.5,4) circle (6pt);
			\fill (-1.5,4) circle (6pt);
			\fill (4.5,4) circle (6pt);
			\fill (1.5,4) circle (6pt);
			\fill (.5,6) circle (6pt);
			\fill (2.5,6) circle (6pt);

			\draw[thick] (0,0) -- (-3,2);
			\draw[thick] (0,0) -- (3,2);
			\draw[thick] (-3,2) -- (-4.5,4);
			\draw[thick] (-3,2) -- (-1.5,4);
			\draw[thick] (3,2) -- (4.5,4);
			\draw[thick] (3,2) -- (1.5,4);
			\draw[thick] (1.5,4) -- (0.5,6);
			\draw[thick] (1.5,4) -- (2.5,6);

		\end{tikzpicture}
	\end{center}
	\caption{A full binary tree with 5 leaves}\label{fig:5_leaves_tree}
\end{figure}

\subsection{Nested integrals as binary trees}\label{sec:integrals_and_trees}

~\cite{Iserles,IserlesTree} present an encoding of integration and commutation into graph theory; more precisely, by graphically representing integrated nested commutators in terms of binary trees.
The construction works as follows:
\begin{enumerate}[label=(\alph*)]
	\item Integration of an operator: appending a single node below the root of the associated tree, which becomes the new root;
	\item Commutation between two operators: integrate the right sub-tree as per rule~(a), then join the two roots with a new node which becomes the new root;
\end{enumerate}
All of the more complex forms of high-order integrated commutators can be expressed in this language.
We will denote by $\mc T_n$ the set of full binary trees with $n$ leaves.
We refer the reader to~\cref{appen_example_3rd_M} for the explicit construction of an example~$\mc T_3$.

\begin{figure}[ht]
	\captionsetup[subfigure]{skip=15pt}
	\begin{subfigure}{0.38\textwidth}
		\centering
		\begin{tikzpicture}
			\begin{scope}[xscale=0.42, yscale=0.40]
				\fill (0,0) circle (6pt);
				\node at (0,2) [above] {$\tau_1$};
				\draw[thick] (0,0) -- (0,2);
			\end{scope}
			\node [align=left, anchor=west] at (0.5,  .75) {$\sim$};
			\node [align=left, anchor=west] at (1.25,  .75) {$\int_0^\kappa H_{\tau_1} dt_1$};
		\end{tikzpicture}
		\caption{rule 1: integration of a tree~$\tau_1$.}
		\label{integration_tree}
	\end{subfigure}
	\begin{subfigure}{0.65\textwidth}
		\centering
		\begin{tikzpicture}
			\begin{scope}[xscale=0.42, yscale=0.40]
				\fill (0,0) circle (6pt);
				\fill (2,2) circle (6pt);
				\draw[thick] (0,0) -- (2,2);
				\draw[thick] (0,0) -- (-2,2);
				\draw[thick] (2,2) -- (2,3.5);
				\node at (-2,2) [above] {$\tau_1$};
				\node at (2,3.5) [above] {$\tau_2$};
			\end{scope}
			\node [align=left, anchor=west] at (1.5, .75) {$\sim$};
			\node [align=left, anchor=west] at (2.25, .75) {$\left[H_{\tau_1}, \int_0^\kappa H_{\tau_2} dt_2 \right]$};
		\end{tikzpicture}
		\caption{rule 2: commutator of two subtrees $\tau_1$ and $\tau_2$.}
		\label{fig_int_and_comm}
	\end{subfigure}
	\caption{Constructing higher-order nested commutators with graphical rules.}
\end{figure}

Note the planar trees generated by the steps described above are not full binary trees, since their nodes do not have exactly 0 or 2 successors.
However, by starting with an enumerated list of true binary trees, we can construct all the terms in the Magnus expansion by attaching an additional node to the right node of each branch to signify integration.
An additional node is also appended to the root of any tree, representing the outermost integration with the total run time~t as the end point.
From now on, we will use the term ``tree'' interchangeably to refer to these quasi-binary trees and to true binary trees.

An arbitrary tree $\tau \in \mc T_n$ has a unique representation -- shown in \cref{fig_tree_rep} -- in terms of a \emph{left-ordered} subtree structure. That is, $\tau$ can be considered as a single left-branch where subtrees $\tau_1,\tau_2,\dots, \tau_r$ are grafted on. Note that this will imply a composition of $n-1$ (the total number of leaves apart from the leftmost one) into $r$ parts (the number of subtrees).
This representation is important in our construction, as the recursive formulas we will derive will refer to this unique way of characterizing a binary tree. We will write $\tau=(\tau_1,\tau_2,\tau_3,...,\tau_r)$ when referring to this structure.

\begin{figure}[htbp]
	\captionsetup{skip=15pt}
	\centering
	\begin{tikzpicture}[scale=0.345]
		\draw[thick] (0,0) -- (-4,4);
		\draw[dotted, thick] (-4,4) -- (-7,7);
		\draw[thick]  (-7,7) -- (-9,9);

		\draw[thick] (0,0) -- (2,2);
		\draw[thick] (-2,2) -- (0,4);
		\draw[thick] (-4,4) -- (-2,6);
		\draw[thick] (-7,7) -- (-5,9);

		\foreach \x/\y in {0/0, -2/2, -4/4, -7/7, 2/2, 0/4, -2/6, -5/9, -9/9} {
			\fill (\x,\y) circle (6pt);}

		\draw[thick] (2,2) -- (2,3.5);
		\node at (2,3.5) [above] {$\tau_1$};
		\draw[thick] (0,4) -- (0,4+1.5);
		\node at  (0,4+1.5)[above] {$\tau_2$};
		\draw[thick] (-2,6) --(-2,6+1.5);
		\node at  (-2,6+1.5) [above] {$\tau_3$};
		\draw[thick] (-5,9) --(-5,10.5);
		\node at  (-5,10.5) [above] {$\tau_r$};

		\fill (0,-1.5) circle (6pt);
		\draw[thick] (0,-1.5) -- (0,0);
	\end{tikzpicture}
	\caption{Every tree $\tau$ has a unique left-ordered representation with grafted sub-trees, $\tau=(\tau_1,\tau_2,\tau_3,...,\tau_r)$.}
	\label{fig_tree_rep}
\end{figure}

\subsection{The Magnus expansion with binary trees}\label{sec:ME_and_trees}

In this binary tree representation, the Magnus expansion can be written as:
\begin{equation}\label{eqn_Magnus_tree}
\mathcal{M}(t) = \sum_{n=1}^\infty M_n(t) = \sum_{n=1}^\infty (-\i)^n \sum_{\tau \in \mathcal{T}_{n}} \alpha_\tau \int_{0}^t H_\tau (\kappa )d \kappa.
\end{equation}
Given a tree $\tau\in\mc T_n$, $H_\tau$ corresponds to a $(n-1)$-fold integral of a $(n-1)$-nested commutator involving $n$ copies of the original time-dependent Hamiltonian $H(t)$, where the details of the structure of the commutator and the boundaries of the integration are given by the structure of~$\tau$.
To exemplify how one arrives at the usual expressions for the Magnus terms, \cref{appen_example_3rd_M} demonstrates how the third Magnus term in~\cref{eq:3rd_magnus_1} can be graphically represented via the binary tree expansion.

The coefficients $\alpha_\tau$ in~\cref{eqn_Magnus_tree} can be evaluated recursively using the unique representation of the sub-trees $\tau=(\tau_1,\tau_2,\dots, \tau_r)$:
\begin{equation}\label{eqn coeff recursive}
\alpha_{(\tau_1,...,\tau_r)} = \frac{B_{r}^+}{r!} \prod_{i=1}^r \alpha_{\tau_i} \  ,
\end{equation}
where $B_r^+$ are the Bernoulli numbers (with $B_1 = 1/2$).
Recalling that  $B_r=0$  for all odd $r \geq 3$, we note that all binary trees with an odd number of grafted sub-trees (except the one having a single sub-tree) have a vanishing $\alpha_{(\tau_1,...,\tau_r)}$ coefficient and hence do not appear in the Magnus expansion.
In this formalism, $H_\tau$ is a left-ordered commutator having as argument the integral of each $H_{\tau_j}$, $j=1,2,\dots,r$, each having the same end point $\kappa$. Namely,
\begin{adjustwidth}{-1cm}{-1cm}
\begin{equation}\label{eq:H_tau}
	H_{\tau=(\tau_1,\tau_2,\dots,\tau_r)}
	=
	\left[\left[\cdots\left[H(\kappa),\int_0^\kappa H_{\tau_r}(s_r)\,\d s_r\right],\int_0^\kappa H_{\tau_{r-1}}(s_{r-1})\,\d s_{r-1}\right],\dots,\int_0^\kappa H_{\tau_1}(s_1)\,\d s_1\right].
\end{equation}
\end{adjustwidth}
which yields a recursive dependence.
This property is leveraged in the next section to compute the coefficients generated by the nested integrals in the Magnus expansion.

\subsection{The integral coefficient}\label{sec:integral_coefficient}

In this section, we analyse how the nested integration associated with a higher-order commutator can be formulated in terms of left-ordered sub-trees (\cref{fig_tree_rep}), in the same way as the $\alpha$ coefficients in \cref{eqn coeff recursive}.
This construction isolates the combinatorial contribution of the nested integrals and is the key ingredient in the subsequent recursion.

The first step from the tree expansion in \cref{eqn_Magnus_tree} is
\begin{equation}\label{eq:M_N_bound}
	\norm{M_n(t)}_\mathrm{op}
	\leq
	\sum_{\tau \in \mathcal{T}_{n}} \abs{\alpha_\tau}
	\norm{\int_{0}^t H_\tau(\kappa)\d\kappa}_\mathrm{op}.
\end{equation}
We now study the scalar nested integral associated with a tree $\tau$, obtained by replacing the operator-valued integrand by~$1$ while retaining the same integration domain. We denote this quantity by $\mc I(\tau;t)$.

Since $\mc I(\tau;t)$ is a nested integration solely of powers of the integration variable, this object results in an inverse product of natural numbers multiplying $t^n$. That is, for every $\tau\in\mc T_n$ we can write
\begin{equation}\label{eq:integral_to_mu}
	\mc I(\tau;t)
	=
	\mu(\tau)t^n
\end{equation}
for some coefficient $\mu(\tau)$ that we call the \emph{integral coefficient}.
We now show that this coefficient satisfies the same left-ordered recursion as the integration structure of the tree.

\begin{lemma}[Integral coefficient recursion formula]\label{lemma:Integral_recursion}
	The \emph{integral coefficient} of a binary tree $\tau \in \mc T_n$ with sub-tree structure $(\tau_1,\dots,\tau_r)$ is obtained by
	\begin{equation}\label{eq:integral_recursion}
		\mu(\tau)
		=
		\frac{1}{n}\mu(\tau_1)\mu(\tau_2)\cdots\mu(\tau_r).
	\end{equation}
\end{lemma}

\begin{proof}
	From \cref{eq:H_tau}, we note that the $r$ different integrals associated with $\tau_1,\tau_2,\dots,\tau_r$ are pairwise independent with respect to each other, and all have the same endpoint of integration, denoted by~$\kappa$, which is precisely the variable of the root integral.
	Thus the scalar nested integral factorises as
	\begin{align}
		\mc I(\tau=(\tau_1,\dots,\tau_r);t)
		&=
		\int_{0}^t \d\kappa \,
		\mc I(\tau_1;\kappa)\mc I(\tau_2;\kappa)\cdots \mc I(\tau_r;\kappa)
		\\
		&=
		\left(\prod_{i=1}^r \mu(\tau_i)\right)
		\int_0^t \kappa^{j_1+\dots+j_r}\d\kappa
		\\
		&=
		\left(\prod_{i=1}^r \mu(\tau_i)\right)
		\int_0^t \kappa^{n-1}\d\kappa
		\\
		&=
		\frac{1}{n}
		\left(\prod_{i=1}^r \mu(\tau_i)\right)
		t^n .
	\end{align}
	Here $j_i$ denotes the number of leaves of the sub-tree $\tau_i$, so that $j_1+\dots+j_r=n-1$.
	Comparing with the definition $\mc I(\tau;t)=\mu(\tau)t^n$ in \cref{eq:integral_to_mu} yields \cref{eq:integral_recursion}.
\end{proof}

The integral coefficients $\mu(\tau)$ defined in \cref{lemma:Integral_recursion} characterise the tree-dependent part of the majorant of the nested commutator $H_\tau$ under the standard bound $\norm{[A,B]} \leq 2 \norm{A}\norm{B}$. We also introduce $A(t)$ as the integrated norm of the generator.

\begin{lemma}[Integral-norm tree bound]\label{lemma:integral_norm_tree_bound}
	Let
	\begin{equation}\label{eq:A_t_definition_local}
		A(t)
		\coloneqq
		\int_0^t \norm{H(s)}_\mathrm{op}\d s .
	\end{equation}
	Then, for every $\tau\in\mc T_n$,
	\begin{equation}\label{eq:integral_norm_tree_bound}
		\norm{\int_0^t H_\tau(\kappa)\d\kappa}_\mathrm{op}
		\leq
		2^{n-1}\mu(\tau)A(t)^n .
	\end{equation}
\end{lemma}

\begin{proof}
	The proof is by induction on the number of leaves. The claim is immediate for the one-leaf tree. For $\tau=(\tau_1,\dots,\tau_r)$, let $j_i$ be the number of leaves of $\tau_i$, so that $\sum_i j_i=n-1$. Using the commutator inequality and the induction hypothesis,
	\begin{align}
		\norm{\int_0^t H_\tau(\kappa)\d\kappa}_\mathrm{op}
		&\leq
		2^r \int_0^t \norm{H(\kappa)}_\mathrm{op} \, \prod_{i=1}^r \norm{\int_0^\kappa H_{\tau_i}(s_i)\,\d s_i}_\mathrm{op}\d\kappa \\
		&\leq
		2^r \left( \prod_{i=1}^r 2^{j_i-1}\mu(\tau_i) \right)
		\left( \int_0^t \norm{H(\kappa)}_\mathrm{op}A(\kappa)^{n-1}\d\kappa \right) \\
		&=
		2^{r+(n-1)-r}\left( \prod_{i=1}^r \mu(\tau_i) \right)
		\left( \int_0^t A'(\kappa)A(\kappa)^{n-1}\d\kappa \right) \\
		&=
		2^{n-1}\left( \prod_{i=1}^r\mu(\tau_i)\right)\frac{1}{n}A(t)^n \\
		&=
		2^{n-1}\mu(\tau)A(t)^n .
	\end{align}
	where we used $A'(\kappa)=\norm{H(\kappa)}_\mathrm{op}$ in the third line, and \cref{eq:integral_recursion} for the last one.
\end{proof}

\begin{example}
	As an example, consider the tree with 7 leaves from \cref{example_tree}. We shall first translate it into the integral
	\begin{adjustwidth}{-2cm}{-2cm}
		{\small
			\begin{equation}
				\int_{0}^t H_\tau (\kappa )d \kappa  = \int_0^t d \kappa
				\Big[ \big[H(\kappa), \int_0^{\kappa} dt_2 [H(t_2), \int_0^{t_2}  dt_3 H(t_3)] \big],\int_0^{\kappa} dt_4\big[ [H(t_4), \int_0^{t_4}dt_5 [H(t_5), \int_0^{t_5} dt_6 H(t_6)]], \int_0^{t_4} dt_7 H(t_7)\big] \Big] \, .
		\end{equation} }
	\end{adjustwidth}
	The scalar volume can directly be computed as 
	\begin{align}
		\mc I(\tau;t) 
		&= 
		\int_0^t d \kappa \int_0^{\kappa} dt_2 \int_0^{t_2} dt_3 \int_0^{\kappa} dt_4 \int_0^{t_4} dt_5 \int_0^{t_5} dt_6 \int_0^{t_4} dt_7 \\
		&=
		\int_0^t d \kappa \, \frac{\kappa^2}{2}  \frac{\kappa^4}{4\cdot 2} \\
		&=
		\frac{1}{7\cdot 4 \cdot 2 \cdot 2} t^7  \label{eq:example_coeff} \\
		&=
		\frac{1}{112} t^7,
	\end{align}
	and related directly to $\mu(\tau)$ by the relation in \cref{eq:integral_to_mu}. 

	The same coefficient can be recovered from the formula in \cref{lemma:Integral_recursion}, 
	\begin{equation}
		\mu(\tau) = \frac{1}{7} \mu(\tau_1) \mu(\tau_2),
	\end{equation}
	where again we can decompose into subtrees $\tau_1=(\chi_0, \chi_1)$ and $\tau_2=\chi_1$. With another recursion we have $\mu(\tau_1)= \frac{1}{4}  \frac{1}{2}$, and $\mu(\tau_2)=\mu(\chi_1)=\frac{1}{2}$. 
	We thus obtain $\mu(\tau) = \frac{1}{7} \frac{1}{4}  \frac{1}{2} \frac{1}{2}$ which matches the coefficient computed by direct calculation in \cref{eq:example_coeff}.
	
	\medskip
	
	We can then apply \cref{lemma:integral_norm_tree_bound} to obtain
	\begin{equation}
		\norm{\int_0^t H_\tau(\kappa)\d\kappa}_\mathrm{op}
		\leq
		\frac{2^{7-1}}{112} A(t)^7
		=
		\frac 4 7 A(t)^7 .
	\end{equation}
\end{example}

\begin{figure}
	\begin{center}
		\begin{tikzpicture}[scale=0.45]

			\fill  (1.25,0) circle (6pt);
			\draw[thick](1.25,0) -- (6.5,2);
			\draw[thick](1.25,0) -- (-4,2);

			\begin{scope}[shift={(2.5,0)}]
				\draw[thick] (4,2) -- (4,3.5);
				\draw[thick] (4,3.5) -- (6,5.25);
				\draw[thick] (6, 5.25) -- (6, 6.75);

				\draw[thick] (4,3.5) -- (0.5,7);
				\draw[thick] (2.25, 5.25) -- (4,7);
				\draw[thick] (4,7) -- (4,8.5);
				\draw[thick] (4,8.5)-- (5.5,10);
				\draw[thick] (4,8.5)-- (2.5,10);
				\draw[thick] (5.5,10) -- (5.5,11.25);

				\fill  (4,2) circle (6pt);
				\fill  (4,3.5) circle (6pt);
				\fill (2.25, 5.25) circle (6pt);
				\fill (6, 5.25) circle (6pt);
				\fill (0.5,7) circle (6pt);
				\fill (4,7) circle (6pt);
				\fill (6, 6.75) circle (6pt);
				\fill (4,8.5) circle (6pt);
				\fill (2.5,10) circle (6pt);
				\fill (5.5,10) circle (6pt);
				\fill (5.5,11.25) circle (6pt);
			\end{scope}

			\draw[thick] (-4,2) -- (-1.5,3.25);
			\draw[thick] (-4,2) -- (-6.5,3.25);
			\draw[thick]  (-1.5,3.25) -- (-1.5,4.75);
			\draw[thick] (-1.5,4.75) -- (0.25, 6.0);
			\draw[thick] (0.25, 6.0) -- (0.25, 7.5);
			\draw[thick] (-1.5,4.75) -- (-3.25, 6.0);

			\fill (-4,2) circle (6pt);
			\fill (-1.5,3.25) circle (6pt);
			\fill (-6.5,3.25) circle (6pt);
			\fill (-1.5,4.75) circle (6pt);
			\fill (0.25,6.0) circle (6pt);
			\fill (0.25,7.5) circle (6pt);
			\fill (-3.25,6.0) circle (6pt);

			\fill (1.25,-1.5) circle (6pt);
			\draw[thick] (1.25,-1.5) -- (1.25,0);

		\end{tikzpicture}
	\end{center}
	\caption{A 7-leaves binary tree.}\label{example_tree}
\end{figure}

\section{A recursion formula for the combined coefficients}\label{sec:combined_recursion}

Since both the $\alpha$ coefficient and the integral coefficient~$\mu$ follow the same left-ordered tree description, for any given tree we can straightforwardly couple them together and obtain another recursion formula.
Our bound on the Magnus term in \cref{eq:M_N_bound} groups together all commutators of depth~$n$ (corresponding to the set of trees $\mc{T}_n$ with $n$ leaves), hence the quantity of interest is
\begin{equation}\label{def:tree_coefficients}
	\nu_n \coloneqq \sum_{\tau \in \mc{T}_n} \abs{\alpha_\tau} \cdot \mu(\tau)
\end{equation}
We recall that an $(n,r)$-\emph{composition} is a sequence of $r$ positive integers summing to~$n$, and that two different orderings of the same numbers count as two distinct compositions. Then, adhering to the unique representation in \cref{fig_tree_rep}, the recursion formula for $\nu_{n+1}$ can be characterised in terms of all $(n,r)$-compositions of $n$~leaves of the trees in $\mc T_{n+1}$ (that is, excluding the leftmost leaf) split into $r$ subtrees, for all $r=1,\dots, n$.

We now introduce an explicit recursion formula for the tree coefficients obtained from the binary-tree formalism; the same scalar recursion for the majorant of the Magnus--Lie characteristic was derived in \cite{lakos2025} through Lie--algebraic methods.

\begin{theorem}[Tree-coefficient recursion]\label{thm:complete_recursion_formula}
	Define the \emph{tree coefficients} as per \cref{def:tree_coefficients}	and set $\nu_0=0$.
	Then, for all $n\geq 1$, their recursion formula is given by
	\begin{equation}\label{eq:beautiful_recursion}
		(n+1) \cdot  \nu_{n+1}
		=
		\sum_{r=1}^{n} \frac{\abs{B_r}}{r!}
		\sum_{\substack{j_1,\dots,j_r \\ \mathrm{composition}(n,r)}}
		\prod_{i=1}^{r} \nu_{j_i} \ ,
	\end{equation}
	with $\nu_1=1$.
\end{theorem}

\begin{proof}
	We recall \cref{eqn coeff recursive} and \cref{lemma:Integral_recursion} for $\alpha_\tau$ and $\mu(\tau)$ respectively.
	For $n\geq2$, we reorganise the sum over trees with $n$ leaves by first fixing the number $r$ of grafted subtrees in the left-ordered representation, and then summing over all ways of distributing the $n-1$ leaves aside from the leftmost one among them -- that is, over all $(n-1,r)$-compositions $(j_1,\dots,j_r)$.
	
	\begin{align}
		\nu_n
		&=
		\sum_{\tau \in \mc T_n} \abs{\alpha_\tau} \cdot \mu(\tau)
		\\
		&=
		\sum_{r=1}^{n-1}
		\sum_{\substack{j_1,\dots,j_r \\ \mathrm{composition}(n-1,r)}}
		\sum_{\tau_{j_1} \in \mc T_{j_1}, \dots,\tau_{j_r} \in \mc T_{j_r}}
		\frac{\abs{B_r}}{r!}
		\prod_{i=1}^r \abs{\alpha_{\tau_{j_i}}}
		\cdot
		\frac{1}{n}
		\prod_{i=1}^r \mu(\tau_{j_i})
		\\
		&=
		\sum_{r=1}^{n-1} \frac{\abs{B_r}}{n \cdot r!}
		\sum_{\substack{j_1,\dots,j_r \\ \mathrm{composition}(n-1,r)}}
		\prod_{i=1}^{r}
		\left(
		\sum_{\tau_i \in \mc{T}_{j_i}}
		\abs{\alpha_{\tau_i}} \cdot \mu(\tau_i)
		\right)
		\\
		&=
		\sum_{r=1}^{n-1} \frac{\abs{B_r}}{n \cdot r!}
		\sum_{\substack{j_1,\dots,j_r \\ \mathrm{composition}(n-1,r)}}
		\prod_{i=1}^{r} \nu_{j_i} \ . \label{eq:final_eq_recursion}
	\end{align}
	Note that the order of the division of $(n-1)$ matters, and hence compositions instead of partitions are considered.
	Relabelling \cref{eq:final_eq_recursion} by $n \rightarrow n+1$ concludes the proof.
\end{proof}

To extract the scaling behaviour of the tree coefficients, one encodes them into a generating function $f(x)$. The recursion in \cref{eq:beautiful_recursion} then yields a differential equation for $f$. An equivalent generating function appeared in \cite{blanes09} as the inverse of an integral function.
 
\begin{lemma}[Generating function of the tree coefficients (cf. \cite{blanes09} )] \label{lem:tree_generating_function}
	Let $(\nu_n)_{n\geq 1}$ be defined by the recursion of \cref{thm:complete_recursion_formula}, and let
	\begin{equation}\label{eq:def_f_genfun}
		f(x)\coloneqq \sum_{n\geq 1}\nu_n x^n .
	\end{equation}
	Then $f$ converges in a neighbourhood of $x=0$ and is the unique analytic solution of the initial value problem
	\begin{equation}\label{eq:ode_correct}
		f'(x)=\Phi(f(x)),
		\qquad
		\Phi(w)\coloneqq 2+\frac{w}{2}-\frac{w}{2}\cot\!\Big(\frac{w}{2}\Big),
		\qquad
		f(0)=0 ,
	\end{equation}
	where $\Phi$ is extended analytically at $w=0$ by setting $\Phi(0)=1$.
\end{lemma}

We prove the lemma directly from the tree-coefficient recursion.
\begin{proof}
	We start from the recursion \cref{eq:beautiful_recursion}. Multiplying by $x^n$ and summing over $n\geq 1$, we obtain
	\begin{equation}\label{eq:coeff_extraction_identity}
		\sum_{n\geq 1}(n+1)\nu_{n+1}x^n
		=
		\sum_{n\geq 1}\sum_{r=1}^{n}\frac{\abs{B_r}}{r!}
		\sum_{\substack{j_1,\dots,j_r\\ \mathrm{composition}(n,r)}}
		\prod_{i=1}^r \nu_{j_i}\,x^n .
	\end{equation}
	The left-hand side is the derivative of $f$,
	\begin{equation}\label{eq:LHS_ode}
		\sum_{n\geq 1}(n+1)\nu_{n+1}x^n=f'(x)-\nu_1=f'(x)-1 .
	\end{equation}
	For the right-hand side, for each fixed $r\geq 1$ we have
	\begin{equation}\label{eq:composition_to_power}
		\sum_{n\geq 1}
		\left(
		\sum_{\substack{j_1,\dots,j_r\\ \mathrm{composition}(n,r)}}
		\prod_{i=1}^r \nu_{j_i}
		\right)x^n
		=
		\left(\sum_{m\geq 1}\nu_m x^m\right)^r
		=
		f(x)^r .
	\end{equation}
	Hence \cref{eq:coeff_extraction_identity} reduces to
	\begin{equation}\label{eq:ode_from_series}
		f'(x)-1=\sum_{r\geq 1}\frac{\abs{B_r}}{r!}\,f(x)^r .
	\end{equation}
	Using the cotangent expansion formula, one has the identity of analytic functions near $w=0$,
	\begin{equation}\label{eq:cot_identity_abs}
		\sum_{r\geq 1}\frac{\abs{B_r}}{r!}\,w^r
		=
		1+\frac{w}{2}-\frac{w}{2}\cot\!\Big(\frac{w}{2}\Big).
	\end{equation}
	Now combining \cref{eq:ode_from_series} and \cref{eq:cot_identity_abs} yields \cref{eq:ode_correct}.
	
	\medskip
	
	At $x=0$, the generating function \cref{eq:def_f_genfun} gives $f(0)=0$ and $f'(0)=\nu_1=1$, so that the differential equation in \cref{eq:ode_correct} becomes $\Phi(f(0))= \Phi(0)=1$.
	On the other hand, using the expansion $\cot z=z^{-1}-z/3+\mc O(z^3)$, we obtain
	\begin{equation}\label{eq:Phi_at_0}
		\lim_{w\to 0} \Phi(w)
		=
		2 + \lim_{w\to 0} \frac{w}{2} - \lim_{w\to 0}\frac{w}{2}  \left(\frac{2}{w}-\frac{w}{6}+ \mc O(w^3) \right)
		=
		1,
	\end{equation}
	which is then consistent with the differential equation. 
	Hence, after setting $\Phi(0)=1$ the function $\Phi$ is well-defined and analytic near $w=0$.
\end{proof}

\subsection{Scaling law for the tree coefficients}

We now arrive at the core analytic result of the construction. 
The differential equation allows us to recover by separation of variables the positive real singular point $\rho$, previously identified in \cite{Blanes98,moan1998}.
Our novel contribution is to prove that $\rho$ is the \emph{unique} dominant singularity, and that the local expansion of $f$ near $x=\rho$ has a square-root form, which is precisely the structure needed to apply the Flajolet--Odlyzko transfer theorem and extract the all-order scaling of the tree coefficients.
 
\begin{theorem}[Scaling of the tree coefficients]\label{thm:asymptotic_tree_coeffs}
	Let $f$ be the generating function in \cref{lem:tree_generating_function}. Then $f$ has its unique dominant singularity at the positive real point $x=\rho$, where in the notation of \cref{convergence_radius},
	\begin{equation}\label{eq:rho_def}
		\rho
		=
		\int_{0}^{2\pi}\frac{\d u}{\Phi(u)}
		=
		\int_{0}^{2\pi}\frac{\d u}{2+\frac{u}{2}-\frac{u}{2}\cot\!\bigl(\frac{u}{2}\bigr)}
		=
		2\xi .
	\end{equation}
	Near $x=\rho$, the generating function satisfies
	\begin{equation}\label{eq:f_sing_scaled}
		f(x)=2\pi-2\sqrt{\pi\rho}\Bigl(1-\frac{x}{\rho}\Bigr)^{1/2}+ \mc O \Bigl(1-\frac{x}{\rho}\Bigr) ;
	\end{equation}
	thus, there exists a global constant $K_\nu<\infty$ such that
	\begin{equation}\label{eq:tree_coefficients_all_order_bound}
		\nu_n\leq K_\nu\,n^{-3/2}\rho^{-n}
	\end{equation}
	for all $n\geq1$. Equivalently, writing $\rho=2\xi$,
	\begin{equation}\label{eq:tree_coefficients_all_order_bound_xi}
		\nu_n\leq K_\nu\,n^{-3/2}\, 2^{-n}\xi^{-n}.
	\end{equation}
\end{theorem}

\begin{proof}
	Since the differential equation \cref{eq:ode_correct} is autonomous -- its right-hand side depends only on $f(x)$ and not directly on $x$ -- we can separate variables and write $\d x=\frac{\d f}{\Phi(f)}$.
	Integrating along the solution and using the initial condition $f(0)=0$ gives the local inverse relation
	\begin{equation}\label{eq:x_as_inverse}
		x=\int_0^{f(x)}\frac{\d u}{\Phi(u)}.
	\end{equation}
	The inverse relation shows that, along the positive real axis, the first possible singular point occurs when $f$ reaches the first pole of $\frac w 2 \cot(w/2)$, namely at $w=2\pi$.
	For $0<u<2\pi$ we have $\Phi(u)>0$, so the solution increases monotonically along the positive real axis until $f(x)\uparrow 2\pi$; the corresponding value of $x$ is therefore
	\begin{equation}\label{eq:rho_again}
		\rho=\int_{0}^{2\pi}\frac{\d u}{\Phi(u)}.
	\end{equation}
	
	\medskip
	
	We now determine the local form of the singularity in the complex plane. Let
	\begin{equation}\label{eq:F_inverse}
		F(w):=\int_0^w\frac{\d u}{\Phi(u)},
	\end{equation}
	where the primitive is obtained by analytic continuation from the real interval~$[0,2\pi]$. Then \cref{eq:x_as_inverse} reads $x=F(f(x))$, with $F(2\pi)=\rho$. Since
	\begin{equation}\label{eq:inv_Phi_local}
		\frac{1}{\Phi(2\pi-\varepsilon)}
		=
		\frac{\varepsilon}{2\pi}+\mc O(\varepsilon^2),
	\end{equation}
	uniformly for complex $\varepsilon$ in a sufficiently small sector, integration gives
	\begin{equation}\label{eq:F_local}
		F(2\pi-\varepsilon)
		=
		\rho-\frac{\varepsilon^2}{4\pi}+\mc O(\varepsilon^3).
	\end{equation}
	In particular, $F'(2\pi)=0$ and $F''(2\pi)=-1/(2\pi)\neq0$, so the standard singular inversion theorem~\cite{Flajolet2013} implies that the inverse branch selected by $f(x)\uparrow 2\pi$ admits a convergent Puiseux expansion
	\begin{equation}\label{eq:full_Puiseux_expansion}
		f(z)
		=
		2\pi+\sum_{k\geq1}c_k\left(1-\frac{z}{\rho}\right)^{k/2},
		\qquad
		c_1=-2\sqrt{\pi\rho},
	\end{equation}
	in a neighbourhood of $\rho$; in particular, we obtain \cref{eq:f_sing_scaled},
	\begin{equation}\label{eq:f_sing_scaled_proof}
		f(z)
		=
		2\pi
		-
		2\sqrt{\pi\rho}\left(1-\frac{z}{\rho}\right)^{1/2}
		+
		\mc O\!\left(1-\frac{z}{\rho}\right).
	\end{equation}
	
	\medskip
	
	It remains to show that $\rho$ is the unique dominant singularity. Since the coefficients $\nu_n$ are strictly positive, Pringsheim's theorem implies that the radius of convergence is attained on the positive real axis; the construction above therefore shows that it equals $\rho$. Moreover, monotone convergence gives
	\begin{equation}
		\sum_{n\geq1}\nu_n\rho^n
		=
		\lim_{x\uparrow\rho}f(x)
		=
		2\pi ,
	\end{equation}
	hence for $\abs{z}=\rho$ and $z\neq\rho$ the triangle inequality yields
	\begin{equation}\label{eq:strict_triangle_boundary}
		\abs{f(z)}
		<
		\sum_{n\geq1}\nu_n\rho^n
		=
		2\pi.
	\end{equation}
	Thus $f(z)$ stays away from the poles $2\pi k$ of $\Phi$, and the differential equation analytically continues the solution across every point of $\abs{z}=\rho$ except $z=\rho$. Combined with the local inversion above, this proves that $f$ is analytic in a $\Delta$-domain at its unique dominant singularity $\rho$.
	
	\medskip
	
	Having established that $\rho$ is the unique dominant singularity of $f$ and that $f$ is analytic in a $\Delta$-domain at $\rho$, we can apply the Flajolet--Odlyzko transfer theorem~\cite[Theorem~VI.1]{Flajolet2013} to the leading term in \cref{eq:f_sing_scaled_proof},
	\begin{equation}\label{eq:transfer_general}
		[z^n]\left(1-\frac{z}{\rho}\right)^{1/2}
		=
		\frac{\rho^{-n}n^{-3/2}}{\Gamma(-1/2)}
		+
		\mc O\!\left(\rho^{-n}n^{-5/2}\right).
	\end{equation}
	The next term in the Puiseux expansion is proportional to $(1-z/\rho)$, whose coefficients vanish
	\begin{equation}
		[z^n]\left(1-\frac{z}{\rho}\right)=0,
		\qquad n\geq2,
	\end{equation}
	so that the next contributing term in the expansion \cref{eq:full_Puiseux_expansion} is proportional to $(1-z/\rho)^{3/2}$ and contributes $\mc O(n^{-5/2}\rho^{-n})$. Hence
	\begin{equation}\label{eq:tree_coefficients_asymptotic_expansion}
		\nu_n
		=
		\frac{-2\sqrt{\pi\rho}}{\Gamma(-1/2)}
		n^{-3/2}\rho^{-n}
		+
		\mc O\!\left(n^{-5/2}\rho^{-n}\right)
		=
		\sqrt{\rho}\,n^{-3/2}\rho^{-n}
		+
		\mc O\!\left(n^{-5/2}\rho^{-n}\right),
	\end{equation}
	where we used $\Gamma(-1/2)=-2\sqrt{\pi}$. This proves that the sequence $\nu_n\rho^n n^{3/2}$ is bounded, and there exists a global constant $K_\nu<\infty$ for which \cref{eq:tree_coefficients_all_order_bound} holds. The equivalent form \cref{eq:tree_coefficients_all_order_bound_xi} follows from $\rho=2\xi$.
	
\end{proof}


Equation~\eqref{eq:tree_coefficients_asymptotic_expansion} immediately yields:
\begin{corollary}[Sharp asymptotic scaling of the tree coefficients]\label{cor:asymptotic_tree_coeffs}
	In the limit $n\to\infty$,
	\begin{equation}\label{eq:nu_asymptotic_rho}
		\nu_n\sim\sqrt{\rho}\,n^{-3/2}\rho^{-n},
	\end{equation}
	or equivalently, using $\rho=2\xi$,
	\begin{equation}\label{eq:nu_asymptotic_xi}
		\nu_n\sim\sqrt{2\xi}\,n^{-3/2}2^{-n}\xi^{-n}.
	\end{equation}
\end{corollary}

The Flajolet--Odlyzko transfer theorem applied to the square-root singularity of $f$ guarantees that the all-order scaling law for $\nu_n$ in \cref{thm:asymptotic_tree_coeffs} is optimal in its dependence on $n$.
Moreover, \cref{cor:asymptotic_tree_coeffs} identifies the exact leading constant in the asymptotic regime:
	\begin{equation}
		\lim_{n\to\infty}\nu_n\rho^n n^{3/2}
		=
		\sqrt{\rho}
		\approx
		1.47436 .
	\end{equation}

\section{Bounding arbitrary Magnus terms}\label{sect:bound}

We now have  the ingredients to prove our per-term main result,

\begin{proof}[Proof of \cref{lemma:bound_magnus_term}]	
	Using \cref{lemma:integral_norm_tree_bound}, we have
	\begin{align}
		\norm{M_n(t)}_{\mathrm{op}}
		&\leq
		\sum_{\tau\in\mc T_n}\abs{\alpha_\tau}
		\norm{\int_0^t H_\tau(\kappa)\,\d\kappa}_{\mathrm{op}} \\
		&\leq
		2^{n-1}A(t)^n\sum_{\tau\in\mc T_n}\abs{\alpha_\tau}\mu(\tau) \\
		&=
		2^{n-1}A(t)^n\nu_n . \label{eq:Magnus_majorant}
	\end{align}
	Inserting the all-order tree-coefficient bound for $\nu_n$ given in \cref{eq:tree_coefficients_all_order_bound_xi} from \cref{thm:asymptotic_tree_coeffs}, we obtain
	\begin{equation}
		\norm{M_n(t)}_{\mathrm{op}}
		\leq
		\frac{K_\nu}{2}\frac{\left(A(t)/\xi\right)^n}{n^{3/2}} .
	\end{equation}
	Thus \cref{eq:Magnus_main_bound} follows with the simple halving
	\begin{equation}\label{eq:KM_Knu_relation}
		K_{\mathcal{M}}=\frac{K_\nu}{2}.
	\end{equation}
\end{proof}

From this per-term bound, we arrive directly at the proof of the main theorem.

\begin{proof}[Proof of \cref{Main_Thm}]
 	Our main result follows directly from applying \cref{lemma:bound_magnus_term} to all terms in the truncation of $\mathcal{M}(t) - \mathcal{M}^{(N)}(t) = \sum_{n\geq N+1} M_n (t)$. 
	Since $n\geq N+1$, we can bound all denominators uniformly by
	\begin{equation}\label{eq:uniform_denominator}
		\frac{1}{n^{3/2}}
		\leq
		\frac{1}{(N+1)^{3/2}},
	\end{equation}
	which allows us to factor them out of the sum. We then use the closed form of the tail of the geometric series in powers of $A(t)/\xi<1$, obtaining
	\begin{equation}
		\norm{\mathcal{M}(t)-\mathcal{M}^{(N)}(t)}_{\mathrm{op}}
		\leq
		\frac{K_{\mathcal{M}}}{(N+1)^{3/2}}
		\frac{\left(A(t)/\xi\right)^{N+1}}{1-A(t)/\xi}.
	\end{equation}
\end{proof}

The corollary follows from the following step:

\begin{proof}[Proof of \cref{cor:propagator_bound}]
	For two operators $X$ and $Y$, the difference of their exponentials has the integral representation
	\begin{equation}\label{eq:exponential_difference_identity}
		\exp(X)-\exp(Y)
		=
		\int_0^1\exp\!\left[(1-s)X\right](X-Y)\exp(sY)\,\d s .
	\end{equation}
	Setting $X=\mathcal{M}(t)$ and $Y=\mathcal{M}^{(N)}(t)$, and using the anti-hermitian property of $\mc M(t)$ and its truncations $\mathcal{M}^{(N)}(t)$, we immediately obtain
	\begin{align}
		\norm{U(t)-U_N(t)}_\mathrm{op}
		&\leq
		\int_0^1
		\norm{\exp\!\left[(1-s)\mathcal{M}(t)\right]}_\mathrm{op}
		\norm{\mathcal{M}(t)-\mathcal{M}^{(N)}(t)}_\mathrm{op}
		\norm{\exp\!\left[s\, \mathcal{M}^{(N)}(t)\right]}_\mathrm{op}\,\d s \\
		&=
		\norm{\mathcal{M}(t)-\mathcal{M}^{(N)}(t)}_\mathrm{op}.
	\end{align}
\end{proof}

The bound can be further tightened by avoiding the uniform substitution of the denominator made in \cref{eq:uniform_denominator} to exploit the closed form of a geometric series.
We retain instead the larger denominator~$m^{3/2}$ inside the sum, and run over all~$m$ above the truncation order~$N$.
This yields a sharper bound:

\begin{corollary}[Sharper truncation bound]\label{cor:stronger_bound}
	Using the notation of \cref{Main_Thm},
	\begin{equation}\label{eq:main_truncation_bound_sum}
		\norm{\mathcal{M}(t)-\mathcal{M}^{(N)}(t)}_\mathrm{op}
		\leq
		K_{\mathcal{M}}\sum_{m\geq N+1}\frac{\left(A(t)/\xi\right)^m}{m^{3/2}},
	\end{equation}
	which converges for $A(t)/\xi\leq1$. 
	In the Hamiltonian case, the same right-hand side also bounds the propagator error,
	\begin{equation}\label{eq:propagator_truncation_bound_sum}
		\norm{U(t)-U_N(t)}_\mathrm{op}
		\leq
		K_{\mathcal{M}}\sum_{m\geq N+1}\frac{\left(A(t)/\xi\right)^m}{m^{3/2}}.
	\end{equation}
\end{corollary}

\section{Direct computation of the first coefficients and numerical prefactors}\label{sec:compute_trees}

Our explicit computation of the \emph{exact} coefficients~$\nu_n$ from the recursion formula \cref{eq:beautiful_recursion} is presented in \cref{table:tree_coefficients} for $n$ up to~10.
The first five coefficients were reported in \cite{blanes09} and attributed to \cite{moan2002backward} (obtained there from the expansion of an unevaluated integral formula) and coincide with ours.
To illustrate the validity of our bound, we compute the coefficients directly from the exact recursion in \cref{eq:beautiful_recursion} and compare them against the rounded bound $\nu_n\leq 4\,n^{-3/2}(2\xi)^{-n}$.

\medskip

The asymptotic result of \cref{thm:asymptotic_tree_coeffs} establishes
\begin{equation}\label{eq:numerical_asymptotic_prediction}
	\nu_n \sim \sqrt{2\xi}\,n^{-3/2}(2\xi)^{-n},
	\qquad
	\xi\approx 1.086868701864 .
\end{equation}
\cref{tab:magnus_coeffs} reports the exact values with the asymptotic scaling law up to $n=100$. In particular, we consider the ratio
\begin{equation}\label{eq:ratio_definition_numerics}
	\frac{\nu_n}{\sqrt{2\xi}\,n^{-3/2}(2\xi)^{-n}},
\end{equation}
which after the pre-asymptotic maximum at $n=3$ decreases towards~$1$ over the computed range, in agreement with the prediction in \cref{eq:nu_asymptotic_xi}.
This observation is supported by the following eventual monotonicity result.

\begin{lemma}[Eventual monotonicity of the normalised majorant]\label{lem:monotonicity}
	The sequence
	\begin{equation}\label{eq:normalised_Magnus_majorant}
		2^{n-1}\xi^n n^{3/2}\nu_n
	\end{equation}
	is eventually strictly decreasing and converges to $\sqrt{2\xi}/2$ from above.
\end{lemma}

\begin{proof}
	Let $\varepsilon=2\pi-f(x)$ and define
	\begin{equation}\label{def:X_capsule}
		X\coloneqq\left(1-\frac{x}{2\xi}\right)^{1/2}. 
	\end{equation}
	As in the proof of \cref{thm:asymptotic_tree_coeffs}, we start by \cref{eq:x_as_inverse}, namely,
	\begin{equation}\label{eq:rho_minus_x_mono}
		\rho-x=\int_{0}^{\varepsilon}\frac{\d s}{\Phi(2\pi-s)}.
	\end{equation}
	Expanding the integrand at $s=0$ and keeping explicitly orders up to $s^3$ gives
	\begin{equation}
		\frac{1}{\Phi(2\pi-s)}
		=
		\frac{s}{2\pi} -\frac{\pi+1}{4\pi^2}s^2 +\frac{4\pi^2+9\pi+3}{24\pi^3}s^3
		+\mc O(s^4);
	\end{equation}
	inserting this in \cref{eq:rho_minus_x_mono} and integrating in the interval from 0 to $\varepsilon$ yields
	\begin{equation}\label{eq:rho_minus_x_refined}
		\rho-x
		=
		\frac{\varepsilon^2}{4\pi} -\frac{\pi+1}{12\pi^2}\varepsilon^3 +\frac{4\pi^2+9\pi+3}{96\pi^3}\varepsilon^4
		+\mc O(\varepsilon^5).
	\end{equation}
	\Cref{def:X_capsule} can be rewritten as $\rho-x=\rho X^2$, thus \cref{eq:rho_minus_x_refined} becomes
	\begin{equation}
		\rho X^2
		=
		\frac{\varepsilon^2}{4\pi}
		\left(
		1-\frac{\pi+1}{3\pi}\varepsilon
		+\frac{4\pi^2+9\pi+3}{24\pi^2}\varepsilon^2
		+\mc O(\varepsilon^3)
		\right);
	\end{equation}
	we invert this relation by writing the ansatz
	\begin{equation}\label{eq:varepsilon_expansion}
		\varepsilon=aX+bX^2+cX^3+\mc O(X^4),
	\end{equation}
	 and then
	\begin{equation}
		\varepsilon^2=a^2X^2+2abX^3+(b^2+2ac)X^4+\mc O(X^5),
	\end{equation}
	\begin{equation}
		\varepsilon^3=a^3X^3+3a^2bX^4+\mc O(X^5),
		\qquad
		\varepsilon^4=a^4X^4+\mc O(X^5).
	\end{equation}
	Substitution into \cref{eq:rho_minus_x_refined} and comparing with $\rho X^2$ gives the system of equations
	\begin{equation}
		\frac{a^2}{4\pi}=\rho,\qquad
		\frac{ab}{2\pi}-\frac{\pi+1}{12\pi^2}a^3=0,
	\end{equation}
	and
	\begin{equation}
		\frac{b^2+2ac}{4\pi}
		-\frac{\pi+1}{4\pi^2}a^2b
		+\frac{4\pi^2+9\pi+3}{96\pi^3}a^4=0.
	\end{equation}
	We solve these equations in order yields the coefficients
	\begin{equation}
		a=2\sqrt{2\pi\xi},\qquad
		b=\frac{4\xi(1+\pi)}{3},\qquad
		c=-\frac{(2\xi)^{3/2}(2\pi^2+7\pi-1)}{18\sqrt{\pi}} ;
	\end{equation}
	 that we insert into  \cref{eq:varepsilon_expansion}; we then recall $f(x) = 2\pi-\varepsilon$ to retrieve
	\begin{equation}\label{eq:refined_Puiseux_expansion}
		f(x)
		=
		2\pi-2\sqrt{2\pi\xi}\,X-\frac{4\xi(1+\pi)}{3}X^2
		+\frac{(2\xi)^{3/2}(2\pi^2+7\pi-1)}{18\sqrt{\pi}}X^3
		+\mc O(X^4).
	\end{equation}
	Since the $X^2 = 1-\frac{x}{2\xi}$ term is linear in $x$, it does not contribute to the asymptotic coefficients for $n\geq 2$. The first two non-polynomial contributions are therefore the $X$ and $X^3$ terms; applying to them the transfer theorem used in the proof of \cref{thm:asymptotic_tree_coeffs} gives respectively
	\begin{equation}
		[x^n]\left(1-\frac{x}{2\xi}\right)^{1/2}
		=
		\frac{(2\xi)^{-n}}{\Gamma(-1/2)}n^{-3/2}
		\left(1+\frac{3}{8n}+\mc O(n^{-2})\right),
	\end{equation}
	and
	\begin{equation}
		[x^n]\left(1-\frac{x}{2\xi}\right)^{3/2}
		=
		\frac{(2\xi)^{-n}}{\Gamma(-3/2)}n^{-5/2}
		\left(1+\mc O(n^{-1})\right).
	\end{equation}
	
	The complete Puiseux expansion established in \cref{eq:full_Puiseux_expansion} allows the transfer theorem to be applied term by term to arbitrary order. In particular, retaining the contributions up to $X^5$ gives
	\begin{equation}
		2^{n-1}\xi^n n^{3/2}\nu_n
		=
		\frac{\sqrt{2\xi}}{2}
		\left(
		1+\frac{q_1}{n}+\frac{q_2}{n^2}+\mc O(n^{-3})
		\right),
	\end{equation}
	for some real constants $q_1$ and $q_2$, where
	\begin{equation}
		q_1=\frac{3}{8}+\frac{\xi(2\pi^2+7\pi-1)}{12\pi}>0.
	\end{equation}
	Taking the difference of two consecutive terms of this complete asymptotic expansion yields
	\begin{equation}
		2^n\xi^{n+1}(n+1)^{3/2}\nu_{n+1}
		-
		2^{n-1}\xi^n n^{3/2}\nu_n
		=
		-\frac{\sqrt{2\xi}\,q_1}{2n^2}+\mc O(n^{-3}),
	\end{equation}
	which is strictly negative for all sufficiently large $n$. The positive $q_1/n$ correction also shows that the limiting value $\sqrt{2\xi}/2$ is approached from above.
\end{proof}

\medskip

It is useful to distinguish the asymptotic and pre-asymptotic contributions to the constant $K_{\mathcal{M}}$ in \cref{lemma:bound_magnus_term}.
\begin{itemize}
	\item \textbf{Asymptotic prefactor.}
	By \cref{eq:nu_asymptotic_xi},
	\begin{equation}\label{eq:asymptotic_constant_numerics}
		\nu_n
		\sim
		\sqrt{2\xi}\,n^{-3/2}\,2^{-n}\,\xi^{-n} ,
	\end{equation}
	hence $K_{\mathcal{M},\infty}=\frac{\sqrt{2\xi}}{2}\approx 0.73718$.
	
	\item \textbf{Numerical pre-asymptotic maximum.}
	Computation of the first $100$ coefficients shows that the ratio in \cref{eq:ratio_definition_numerics} attains its maximum over this range at $n=3$ (where the exact value is $\nu_3= \frac{5}{72}$).
	Together with the eventual monotonicity established in \cref{lem:monotonicity}, this suggests that
	\begin{equation}\label{eq:KM_numerical_peak}
		K_{\mathcal{M},3}
		=
		\frac{1}{2}\nu_3(2\xi)^3 \, 3^{3/2}
		\approx
		1.85315024 
	\end{equation}
	is the smallest global constant for the Magnus bounds we present in the main results. For explicit practical estimates, this motivates the rounded choice $K_{\mathcal{M}}=2$. 
\end{itemize}

\begin{table}[htbp]
	\centering
	\renewcommand{\arraystretch}{1.5}
	\captionsetup{skip=15pt}
	\resizebox{1.03\textwidth}{!}{%
		\begin{tabular}{|c|>{\centering\arraybackslash}m{3cm}|c|c|}
			\hline
			\textbf{Tree coefficient} & \textbf{Exact value} & \textbf{Decimal representation} & \textbf{Rounded bound} \\ \hline
			$\nu_1$  & $1$ & $1.00000000 \times 10^{0}$ & $1.84014867 \times 10^{0}$ \\ \hline
			$\nu_2$  & $\tfrac{1}{4}$ & $2.50000000 \times 10^{-1}$ & $2.99295951 \times 10^{-1}$ \\ \hline
			$\nu_3$  & $\tfrac{5}{72}$ & $6.94444444 \times 10^{-2}$ & $7.49474522 \times 10^{-2}$ \\ \hline
			$\nu_4$  & $\tfrac{11}{576}$ & $1.90972222 \times 10^{-2}$ & $2.23945165 \times 10^{-2}$ \\ \hline
			$\nu_5$  & $\tfrac{479}{86400}$ & $5.54398148 \times 10^{-3}$ & $7.37173293 \times 10^{-3}$ \\ \hline
			$\nu_6$  & $\tfrac{1769}{1036800}$ & $1.70621142 \times 10^{-3}$ & $2.57982737 \times 10^{-3}$ \\ \hline
			$\nu_7$  & $\tfrac{34091}{60963840}$ & $5.59200339 \times 10^{-4}$ & $9.41810185 \times 10^{-4}$ \\ \hline
			$\nu_8$  & $\tfrac{943633}{4877107200}$ & $1.93482112 \times 10^{-4}$ & $3.54624212 \times 10^{-4}$ \\ \hline
			$\nu_9$  & $\tfrac{92107357}{1316818944000}$ & $6.99468651 \times 10^{-5}$ & $1.36720149 \times 10^{-4}$ \\ \hline
			$\nu_{10}$ & $\tfrac{688988827}{26336378880000}$ & $2.61611070 \times 10^{-5}$ & $5.37018451 \times 10^{-5}$ \\ \hline
		\end{tabular}
	}
	\caption{Exact tree coefficients $\nu_n$ for $1\leq n\leq 10$, their decimal values, and the bound $\nu_n\leq 4\,n^{-3/2}(2\xi)^{-n}$ linked to the rounded Magnus prefactor $K_{\mathcal{M}}=2$.}
	\label{table:tree_coefficients}
\end{table}

\begin{figure}[tbp]
	\centering
	\includegraphics[width=0.82\textwidth]{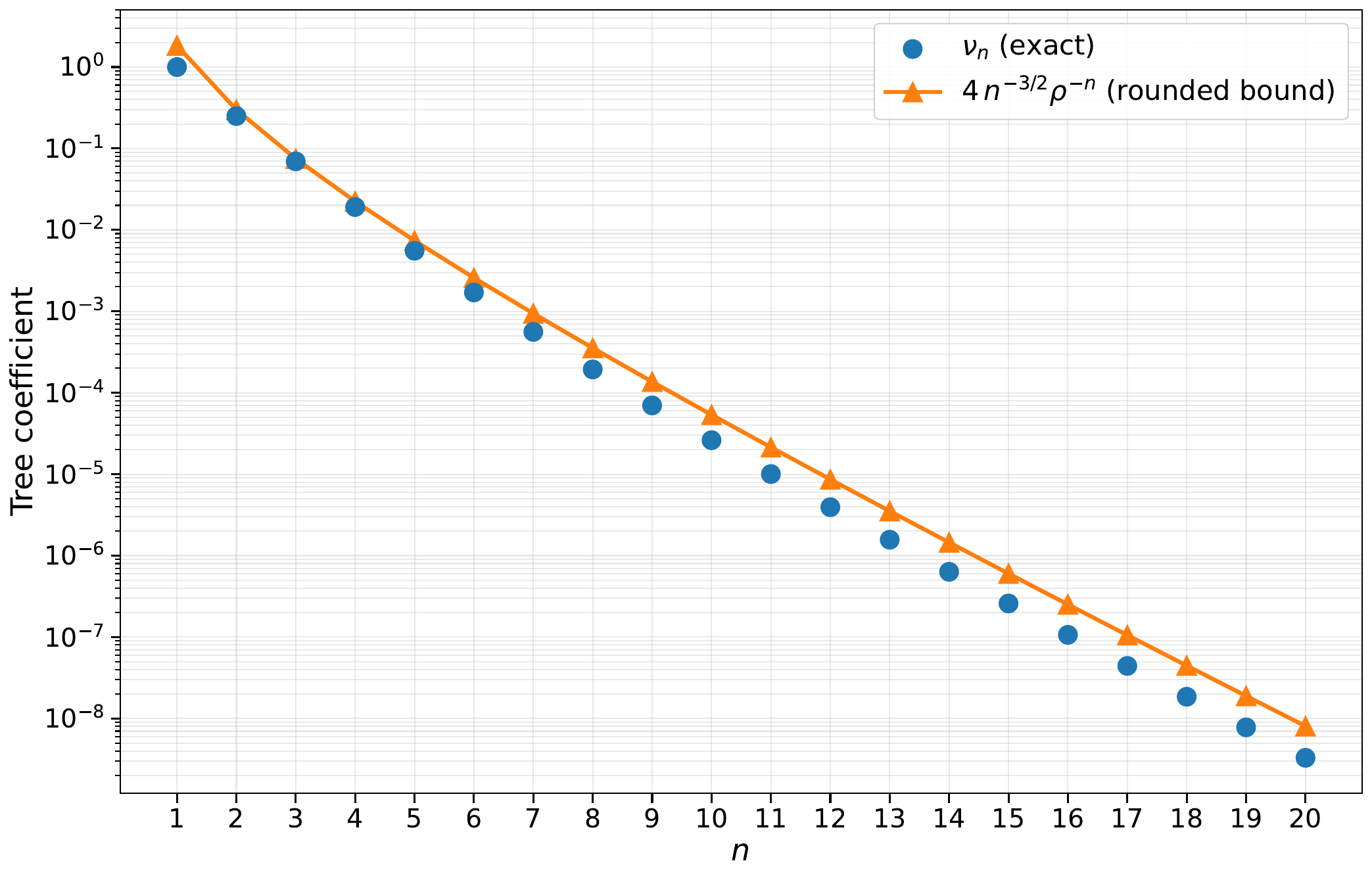}
	\caption{Exact tree coefficients $\nu_n$, computed from the recursion in \cref{eq:beautiful_recursion}, compared with the rounded bound $4\,n^{-3/2}\rho^{-n}$ for $1\leq n\leq 20$. The vertical axis is logarithmic.}
	\label{fig:tree_coefficients_global_bound}
\end{figure}

\FloatBarrier

\section{Discussion}

In this note, we obtain a new bound for an individual Magnus term of any order and derive from it a closed-form bound on truncations of the Magnus expansion.
Compared to previous structure-free estimates, our per-term bound gains an additional factor $n^{-3/2}$, yielding a suppression of $(N+1)^{-3/2}$ in the resulting truncation bound.
We show that this order dependence is optimal for the tree majorant considered in this work, and we also identify the sharp asymptotic prefactor~$K_{\mathcal{M},\infty}=\frac{\sqrt{2\xi}}{2}$, below the previous prefactor $\pi$.
Separately, numerical and theoretical evidence indicate a global rounded value $K_{\mathcal{M}}=2$.
This result provides a simple closed-form formula for controlling the Magnus truncation and, in the Hamiltonian case, the corresponding approximate unitary propagator.

The additional polynomial suppression is particularly relevant for evolution times near the convergence radius and for low-order truncations, where exponential decay has not yet become dominant.
With the rounded prefactor $K_{\mathcal{M}}=2$, our bound is already more than an order of magnitude smaller than the Moan--Blanes bound at $N=4$.
For Hamiltonian simulation, our tighter error analysis can reduce the truncation order required to reach a prescribed precision, thereby reducing the cost of finite-order Magnus implementations~\cite{Fang2025,Moreno2024}.
More generally, the structure-free nature of our analysis makes the bounds directly applicable to time-dependent evolutions in mathematical physics and control theory whenever a rigorous truncation error is required.

An interesting direction for future work is to analytically certify the prefactor in the pre-asymptotic regime and, more broadly, to move beyond the structure-free setting by exploiting locality and the interaction graph of the generators.
Structural information could exclude certain tree configurations a priori through commutator cancellations, potentially yielding sharper bounds and smaller truncation errors.
More broadly, it would be interesting to understand whether analogous tree-coefficient methods can be adapted to open-system generators.
In that setting, the cone of Lindblad operators~\cite{lindblad1976} is not closed under commutation, so finite-order Magnus-type truncations may not guarantee a legitimate quantum evolution at all times and may require additional correction terms.

\subsection*{\large Acknowledgements}
We thank Nikolas Breuckmann and Quirin Vogel for insightful discussions.
HA and TC were supported in part by the EPSRC Prosperity Partnership in Quantum Software for Simulation and Modelling (grant EP/S005021/1), and by the UK Hub in Quantum Computing and Simulation, part of the UK National Quantum Technologies Programme with funding from UKRI EPSRC (grant EP/T001062/1).
EO~has been supported by the ERC under grant agreement no.~101001976 (project~EQUIPTNT).

\printbibliography
\newpage
\appendix

\begin{center}
 {\LARGE \textbf{Appendix} }
 \end{center}

\section{Example: third-order Magnus term from binary trees}\label{appen_example_3rd_M}

First, to exemplify the graphical formulation in terms of `quasi-binary' trees, we enumerate those with at most $10$ total nodes in \cref{fig:tree_table}.

\begin{figure}[t]
	\centering
	\makebox[\textwidth][c]{%
		\includegraphics[width=1.20\textwidth]{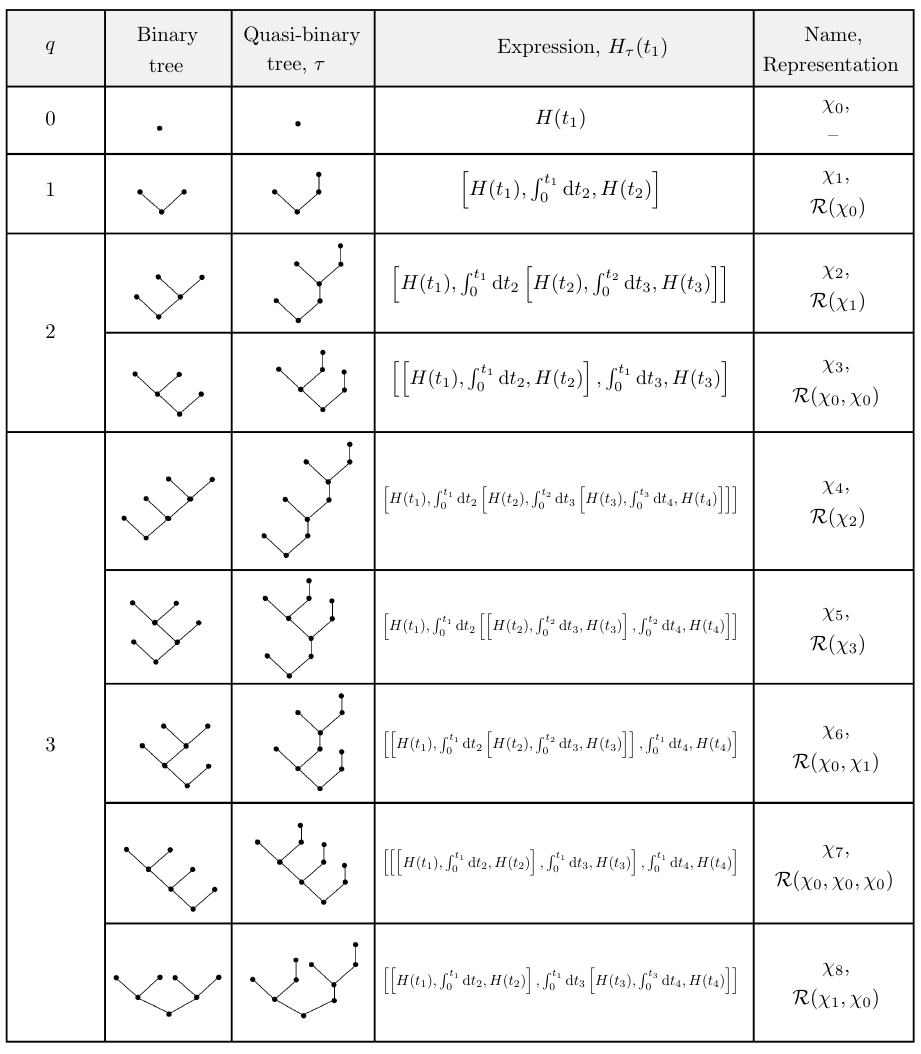}
	}
	\vspace{7pt}
	\caption{First few binary trees with $q$ internal nodes, their respective `quasi-binary' tree, $\tau$, and $H_\tau$.}
	\label{fig:tree_table}
\end{figure}

This appendix demonstrates how the expression for $M_3(t)$ in \cref{eq:3rd_magnus_1} is equivalent to the corresponding term in \cref{eqn_Magnus_tree},
\begin{equation}
	M_3(t)=(-\i)^3 \sum_{\tau\in\mathcal{T}_{3}}\alpha_\tau\int_0^t dt_1\,H_\tau(t_1).
\end{equation}
Using \cref{fig:tree_table}, we find
\begin{align}
	M_3(t)
	&=(-\i)^3 \alpha_{\mathcal{R}(\chi_1)}\int_0^t dt_1
	\left[H(t_1),\int_0^{t_1}dt_2\left[H(t_2),\int_0^{t_2}dt_3H(t_3)\right]\right]
	\notag\\
	&\quad+(-\i)^3 \alpha_{\mathcal{R}(\chi_0,\chi_0)}\int_0^t dt_1
	\left[\left[H(t_1),\int_0^{t_1}dt_2H(t_2)\right],\int_0^{t_1}dt_3H(t_3)\right].
\end{align}
Consistently with \cref{eqn coeff recursive}, we set $\beta_0 = 1$ and $\beta_r =  B_r^+/r!$ so that $\beta_1=1/2$, $\beta_2=1/12$, $\alpha_{\chi_0}=\beta_0$, and $\alpha_{\chi_1}=\beta_1\alpha_{\chi_0}=\beta_0 \beta_1$. The coefficients are then given by
\begin{align}
	\alpha_{\mathcal{R}(\chi_1)}&=\beta_1\alpha_{\chi_1}=\beta_0 \beta_1^2=\frac{1}{4},
	\\
	\alpha_{\mathcal{R}(\chi_0,\chi_0)}&=\beta_2\alpha_{\chi_0}\alpha_{\chi_0}=\beta_0^2 \beta_2=\frac{1}{12}.
\end{align}
Hence we have an expression for the third Magnus term:
\begin{adjustwidth}{-2cm}{-2cm}
\begin{equation}\label{eq:mag_3a}
		M_3(t)
		=
		\frac{(-\i)^3}{4}\int_0^t dt_1\int_0^{t_1}dt_2\int_0^{t_2}dt_3
		\left[H(t_1),\left[H(t_2),H(t_3)\right]\right]
		\, + \; 
		\frac{(-\i)^3}{12}\int_0^t dt_1\int_0^{t_1}dt_2\int_0^{t_1}dt_3
		\left[\left[H(t_1),H(t_2)\right],H(t_3)\right].
\end{equation}
\end{adjustwidth} 
Note that this protocol will generally lead to superficially different expressions for the Magnus terms.
This is explicit here, as \cref{eq:mag_3a} appears different from \cref{eq:3rd_magnus_1}, which we restate here for reference:
\begin{equation}
	M_3(t)
	=\frac{(-\i)^3}{6}\int_0^t dt_1\int_0^{t_1}dt_2\int_0^{t_2}dt_3
	\left(
	\left[\left[H(t_1),H(t_2)\right],H(t_3)\right]
	+\left[H(t_1),\left[H(t_2),H(t_3)\right]\right]
	\right).
\end{equation}
The expression for $M_3$ is thus not unique due to integral relations and the Jacobi identity; below we demonstrate that simple algebra takes us from one to the other.
Taking the second term of \cref{eq:mag_3a}, we split the integration domain
\begin{align}
	&\int_0^t dt_1\int_0^{t_1}dt_2\int_0^{t_1}dt_3
	\left[\left[H(t_1),H(t_2)\right],H(t_3)\right] \\
	&=\int_0^t dt_1\int_0^{t_1}dt_2\int_0^{t_2}dt_3
	\left[\left[H(t_1),H(t_2)\right],H(t_3)\right] \\
	&\quad+\int_0^t dt_1\int_0^{t_1}dt_2\int_{t_2}^{t_1}dt_3
	\left[\left[H(t_1),H(t_2)\right],H(t_3)\right]. \label{eq:splitted_domain}
\end{align}
Using the integral identity
$\int_0^\alpha dy\int_y^\alpha dx\,f(x,y)=\int_0^\alpha dx\int_0^x dy\,f(x,y)$, we reformulate \cref{eq:splitted_domain} as
\begin{align}
	&\int_0^t dt_1\int_0^{t_1}dt_2\int_{t_2}^{t_1}dt_3
	\left[\left[H(t_1),H(t_2)\right],H(t_3)\right] \\
	&=\int_0^t dt_1\int_0^{t_1}dt_3\int_0^{t_3}dt_2
	\left[\left[H(t_1),H(t_2)\right],H(t_3)\right] \\
	&=\int_0^t dt_1\int_0^{t_1}dt_2\int_0^{t_2}dt_3
	\left[\left[H(t_1),H(t_3)\right],H(t_2)\right]
	\label{eqn relabel index}
	\\
	&=\int_0^t dt_1\int_0^{t_1}dt_2\int_0^{t_2}dt_3
	\left(
	\left[\left[H(t_1),H(t_2)\right],H(t_3)\right]
	-\left[H(t_1),\left[H(t_2),H(t_3)\right]\right]
	\right),
	\label{eqn apply jacobi}
\end{align}
where \cref{eqn relabel index} follows from relabelling the integration variables, while \cref{eqn apply jacobi} follows from the Jacobi identity.
Substituting the above relation back into \cref{eq:mag_3a} gives \cref{eq:3rd_magnus_1}.

\FloatBarrier 
\newpage 

\begin{table}[p]
	\centering
	\caption{Comparison of the exact tree coefficients $\nu_n$, computed from the recursion in \cref{eq:beautiful_recursion}, with the sharp asymptotic prediction $\sqrt{\rho}\,n^{-3/2}\rho^{-n}$. The ratio column illustrates convergence towards the Flajolet--Odlyzko asymptotic law. The third column gives the rounded bound $\nu_n\leq 4\,n^{-3/2}\rho^{-n}$ used for the all-order tree-coefficient bound.}
	\label{tab:magnus_coeffs}
	\scriptsize
	\setlength{\tabcolsep}{4pt}
	\renewcommand{\arraystretch}{1.08}
	\begin{tabular}{c c c c c}
		\hline
		\textbf{Order} & \textbf{Exact} & \textbf{Rounded} & \textbf{Asymptotic} & \textbf{Ratio} \\
		$n$ & $\nu_n$ & $4\,n^{-3/2}\rho^{-n}$ & $\sqrt{\rho}\,n^{-3/2}\rho^{-n}$ & $\frac{\nu_n}{\sqrt{\rho}\,n^{-3/2}\rho^{-n}}$ \\
		\hline
		\noalign{\vskip 2pt}
		1 & $1.0000 \times 10^{0}$ & $1.8401 \times 10^{0}$ & $6.7826 \times 10^{-1}$ & $1.474$ \\
		2 & $2.5000 \times 10^{-1}$ & $2.9930 \times 10^{-1}$ & $1.1032 \times 10^{-1}$ & $2.266$ \\
		3 & $6.9444 \times 10^{-2}$ & $7.4947 \times 10^{-2}$ & $2.7625 \times 10^{-2}$ & $2.514$ \\
		4 & $1.9097 \times 10^{-2}$ & $2.2395 \times 10^{-2}$ & $8.2544 \times 10^{-3}$ & $2.314$ \\
		5 & $5.5440 \times 10^{-3}$ & $7.3717 \times 10^{-3}$ & $2.7171 \times 10^{-3}$ & $2.040$ \\
		6 & $1.7062 \times 10^{-3}$ & $2.5798 \times 10^{-3}$ & $9.5090 \times 10^{-4}$ & $1.794$ \\
		7 & $5.5920 \times 10^{-4}$ & $9.4181 \times 10^{-4}$ & $3.4714 \times 10^{-4}$ & $1.611$ \\
		8 & $1.9348 \times 10^{-4}$ & $3.5462 \times 10^{-4}$ & $1.3071 \times 10^{-4}$ & $1.480$ \\
		9 & $6.9947 \times 10^{-5}$ & $1.3672 \times 10^{-4}$ & $5.0394 \times 10^{-5}$ & $1.388$ \\
		10 & $2.6161 \times 10^{-5}$ & $5.3702 \times 10^{-5}$ & $1.9794 \times 10^{-5}$ & $1.322$ \\
		11 & $1.0045 \times 10^{-5}$ & $2.1414 \times 10^{-5}$ & $7.8929 \times 10^{-6}$ & $1.273$ \\
		12 & $3.9373 \times 10^{-6}$ & $8.6458 \times 10^{-6}$ & $3.1867 \times 10^{-6}$ & $1.236$ \\
		13 & $1.5687 \times 10^{-6}$ & $3.5274 \times 10^{-6}$ & $1.3002 \times 10^{-6}$ & $1.207$ \\
		14 & $6.3343 \times 10^{-7}$ & $1.4520 \times 10^{-6}$ & $5.3520 \times 10^{-7}$ & $1.184$ \\
		15 & $2.5861 \times 10^{-7}$ & $6.0231 \times 10^{-7}$ & $2.2200 \times 10^{-7}$ & $1.165$ \\
		16 & $1.0656 \times 10^{-7}$ & $2.5152 \times 10^{-7}$ & $9.2707 \times 10^{-8}$ & $1.149$ \\
		17 & $4.4260 \times 10^{-8}$ & $1.0565 \times 10^{-7}$ & $3.8941 \times 10^{-8}$ & $1.137$ \\
		18 & $1.8509 \times 10^{-8}$ & $4.4609 \times 10^{-8}$ & $1.6442 \times 10^{-8}$ & $1.126$ \\
		19 & $7.7863 \times 10^{-9}$ & $1.8923 \times 10^{-8}$ & $6.9749 \times 10^{-9}$ & $1.116$ \\
		20 & $3.2927 \times 10^{-9}$ & $8.0607 \times 10^{-9}$ & $2.9711 \times 10^{-9}$ & $1.108$ \\
		21 & $1.3989 \times 10^{-9}$ & $3.4465 \times 10^{-9}$ & $1.2704 \times 10^{-9}$ & $1.101$ \\
		22 & $5.9678 \times 10^{-10}$ & $1.4787 \times 10^{-9}$ & $5.4502 \times 10^{-10}$ & $1.095$ \\
		23 & $2.5554 \times 10^{-10}$ & $6.3636 \times 10^{-10}$ & $2.3456 \times 10^{-10}$ & $1.089$ \\
		24 & $1.0979 \times 10^{-10}$ & $2.7465 \times 10^{-10}$ & $1.0123 \times 10^{-10}$ & $1.085$ \\
		25 & $4.7316 \times 10^{-11}$ & $1.1884 \times 10^{-10}$ & $4.3804 \times 10^{-11}$ & $1.080$ \\
		26 & $2.0448 \times 10^{-11}$ & $5.1548 \times 10^{-11}$ & $1.9000 \times 10^{-11}$ & $1.076$ \\
		27 & $8.8594 \times 10^{-12}$ & $2.2409 \times 10^{-11}$ & $8.2597 \times 10^{-12}$ & $1.073$ \\
		28 & $3.8475 \times 10^{-12}$ & $9.7617 \times 10^{-12}$ & $3.5981 \times 10^{-12}$ & $1.069$ \\
		29 & $1.6745 \times 10^{-12}$ & $4.2605 \times 10^{-12}$ & $1.5704 \times 10^{-12}$ & $1.066$ \\
		30 & $7.3027 \times 10^{-13}$ & $1.8628 \times 10^{-12}$ & $6.8661 \times 10^{-13}$ & $1.064$ \\
		31 & $3.1907 \times 10^{-13}$ & $8.1583 \times 10^{-13}$ & $3.0071 \times 10^{-13}$ & $1.061$ \\
		32 & $1.3965 \times 10^{-13}$ & $3.5786 \times 10^{-13}$ & $1.3190 \times 10^{-13}$ & $1.059$ \\
		33 & $6.1220 \times 10^{-14}$ & $1.5720 \times 10^{-13}$ & $5.7943 \times 10^{-14}$ & $1.057$ \\
		34 & $2.6879 \times 10^{-14}$ & $6.9151 \times 10^{-14}$ & $2.5489 \times 10^{-14}$ & $1.055$ \\
		35 & $1.1818 \times 10^{-14}$ & $3.0459 \times 10^{-14}$ & $1.1227 \times 10^{-14}$ & $1.053$ \\
		36 & $5.2032 \times 10^{-15}$ & $1.3432 \times 10^{-14}$ & $4.9510 \times 10^{-15}$ & $1.051$ \\
		37 & $2.2937 \times 10^{-15}$ & $5.9306 \times 10^{-15}$ & $2.1859 \times 10^{-15}$ & $1.049$ \\
		38 & $1.0123 \times 10^{-15}$ & $2.6213 \times 10^{-15}$ & $9.6618 \times 10^{-16}$ & $1.048$ \\
		39 & $4.4730 \times 10^{-16}$ & $1.1598 \times 10^{-15}$ & $4.2749 \times 10^{-16}$ & $1.046$ \\
		40 & $1.9785 \times 10^{-16}$ & $5.1367 \times 10^{-16}$ & $1.8933 \times 10^{-16}$ & $1.045$ \\
		41 & $8.7603 \times 10^{-17}$ & $2.2772 \times 10^{-16}$ & $8.3934 \times 10^{-17}$ & $1.044$ \\
		42 & $3.8825 \times 10^{-17}$ & $1.0104 \times 10^{-16}$ & $3.7242 \times 10^{-17}$ & $1.043$ \\
		43 & $1.7223 \times 10^{-17}$ & $4.4870 \times 10^{-17}$ & $1.6539 \times 10^{-17}$ & $1.041$ \\
		44 & $7.6465 \times 10^{-18}$ & $1.9942 \times 10^{-17}$ & $7.3504 \times 10^{-18}$ & $1.040$ \\
		45 & $3.3977 \times 10^{-18}$ & $8.8700 \times 10^{-18}$ & $3.2694 \times 10^{-18}$ & $1.039$ \\
		46 & $1.5110 \times 10^{-18}$ & $3.9482 \times 10^{-18}$ & $1.4553 \times 10^{-18}$ & $1.038$ \\
		47 & $6.7243 \times 10^{-19}$ & $1.7587 \times 10^{-18}$ & $6.4822 \times 10^{-19}$ & $1.037$ \\
		48 & $2.9947 \times 10^{-19}$ & $7.8390 \times 10^{-19}$ & $2.8894 \times 10^{-19}$ & $1.036$ \\
		49 & $1.3346 \times 10^{-19}$ & $3.4964 \times 10^{-19}$ & $1.2887 \times 10^{-19}$ & $1.036$ \\
		50 & $5.9520 \times 10^{-20}$ & $1.5605 \times 10^{-19}$ & $5.7517 \times 10^{-20}$ & $1.035$ \\
		\hline
	\end{tabular}
\end{table}

\begin{table}[p]
	\centering
	\scriptsize
	\setlength{\tabcolsep}{4pt}
	\renewcommand{\arraystretch}{1.08}
	\begin{tabular}{c c c c c}
		\hline
		\textbf{Order} & \textbf{Exact} & \textbf{Rounded} & \textbf{Asymptotic} & \textbf{Ratio} \\
		$n$ & $\nu_n$ & $4\,n^{-3/2}\rho^{-n}$ & $\sqrt{\rho}\,n^{-3/2}\rho^{-n}$ & $\frac{\nu_n}{\sqrt{\rho}\,n^{-3/2}\rho^{-n}}$ \\
		\hline
		\noalign{\vskip 2pt}
		51 & $2.6560 \times 10^{-20}$ & $6.9686 \times 10^{-20}$ & $2.5685 \times 10^{-20}$ & $1.034$ \\
		52 & $1.1859 \times 10^{-20}$ & $3.1138 \times 10^{-20}$ & $1.1477 \times 10^{-20}$ & $1.033$ \\
		53 & $5.2985 \times 10^{-21}$ & $1.3921 \times 10^{-20}$ & $5.1311 \times 10^{-21}$ & $1.033$ \\
		54 & $2.3686 \times 10^{-21}$ & $6.2271 \times 10^{-21}$ & $2.2953 \times 10^{-21}$ & $1.032$ \\
		55 & $1.0594 \times 10^{-21}$ & $2.7869 \times 10^{-21}$ & $1.0272 \times 10^{-21}$ & $1.031$ \\
		56 & $4.7407 \times 10^{-22}$ & $1.2479 \times 10^{-21}$ & $4.5997 \times 10^{-22}$ & $1.031$ \\
		57 & $2.1225 \times 10^{-22}$ & $5.5904 \times 10^{-22}$ & $2.0606 \times 10^{-22}$ & $1.030$ \\
		58 & $9.5077 \times 10^{-23}$ & $2.5056 \times 10^{-22}$ & $9.2353 \times 10^{-23}$ & $1.029$ \\
		59 & $4.2609 \times 10^{-23}$ & $1.1235 \times 10^{-22}$ & $4.1410 \times 10^{-23}$ & $1.029$ \\
		60 & $1.9104 \times 10^{-23}$ & $5.0397 \times 10^{-23}$ & $1.8576 \times 10^{-23}$ & $1.028$ \\
		61 & $8.5689 \times 10^{-24}$ & $2.2617 \times 10^{-23}$ & $8.3364 \times 10^{-24}$ & $1.028$ \\
		62 & $3.8452 \times 10^{-24}$ & $1.0154 \times 10^{-23}$ & $3.7426 \times 10^{-24}$ & $1.027$ \\
		63 & $1.7262 \times 10^{-24}$ & $4.5604 \times 10^{-24}$ & $1.6809 \times 10^{-24}$ & $1.027$ \\
		64 & $7.7522 \times 10^{-25}$ & $2.0490 \times 10^{-24}$ & $7.5523 \times 10^{-25}$ & $1.026$ \\
		65 & $3.4828 \times 10^{-25}$ & $9.2094 \times 10^{-25}$ & $3.3945 \times 10^{-25}$ & $1.026$ \\
		66 & $1.5653 \times 10^{-25}$ & $4.1407 \times 10^{-25}$ & $1.5262 \times 10^{-25}$ & $1.026$ \\
		67 & $7.0374 \times 10^{-26}$ & $1.8624 \times 10^{-25}$ & $6.8646 \times 10^{-26}$ & $1.025$ \\
		68 & $3.1651 \times 10^{-26}$ & $8.3794 \times 10^{-26}$ & $3.0886 \times 10^{-26}$ & $1.025$ \\
		69 & $1.4240 \times 10^{-26}$ & $3.7714 \times 10^{-26}$ & $1.3901 \times 10^{-26}$ & $1.024$ \\
		70 & $6.4086 \times 10^{-27}$ & $1.6979 \times 10^{-26}$ & $6.2584 \times 10^{-27}$ & $1.024$ \\
		71 & $2.8851 \times 10^{-27}$ & $7.6466 \times 10^{-27}$ & $2.8185 \times 10^{-27}$ & $1.024$ \\
		72 & $1.2992 \times 10^{-27}$ & $3.4447 \times 10^{-27}$ & $1.2697 \times 10^{-27}$ & $1.023$ \\
		73 & $5.8526 \times 10^{-28}$ & $1.5522 \times 10^{-27}$ & $5.7214 \times 10^{-28}$ & $1.023$ \\
		74 & $2.6372 \times 10^{-28}$ & $6.9966 \times 10^{-28}$ & $2.5789 \times 10^{-28}$ & $1.023$ \\
		75 & $1.1886 \times 10^{-28}$ & $3.1545 \times 10^{-28}$ & $1.1627 \times 10^{-28}$ & $1.022$ \\
		76 & $5.3589 \times 10^{-29}$ & $1.4227 \times 10^{-28}$ & $5.2438 \times 10^{-29}$ & $1.022$ \\
		77 & $2.4167 \times 10^{-29}$ & $6.4177 \times 10^{-29}$ & $2.3655 \times 10^{-29}$ & $1.022$ \\
		78 & $1.0901 \times 10^{-29}$ & $2.8958 \times 10^{-29}$ & $1.0674 \times 10^{-29}$ & $1.021$ \\
		79 & $4.9187 \times 10^{-30}$ & $1.3070 \times 10^{-29}$ & $4.8173 \times 10^{-30}$ & $1.021$ \\
		80 & $2.2199 \times 10^{-30}$ & $5.9001 \times 10^{-30}$ & $2.1747 \times 10^{-30}$ & $1.021$ \\
		81 & $1.0021 \times 10^{-30}$ & $2.6642 \times 10^{-30}$ & $9.8198 \times 10^{-31}$ & $1.020$ \\
		82 & $4.5248 \times 10^{-31}$ & $1.2033 \times 10^{-30}$ & $4.4351 \times 10^{-31}$ & $1.020$ \\
		83 & $2.0436 \times 10^{-31}$ & $5.4357 \times 10^{-31}$ & $2.0035 \times 10^{-31}$ & $1.020$ \\
		84 & $9.2314 \times 10^{-32}$ & $2.4561 \times 10^{-31}$ & $9.0529 \times 10^{-32}$ & $1.020$ \\
		85 & $4.1711 \times 10^{-32}$ & $1.1100 \times 10^{-31}$ & $4.0914 \times 10^{-32}$ & $1.019$ \\
		86 & $1.8850 \times 10^{-32}$ & $5.0177 \times 10^{-32}$ & $1.8495 \times 10^{-32}$ & $1.019$ \\
		87 & $8.5208 \times 10^{-33}$ & $2.2686 \times 10^{-32}$ & $8.3620 \times 10^{-33}$ & $1.019$ \\
		88 & $3.8524 \times 10^{-33}$ & $1.0259 \times 10^{-32}$ & $3.7814 \times 10^{-33}$ & $1.019$ \\
		89 & $1.7421 \times 10^{-33}$ & $4.6403 \times 10^{-33}$ & $1.7104 \times 10^{-33}$ & $1.019$ \\
		90 & $7.8792 \times 10^{-34}$ & $2.0992 \times 10^{-33}$ & $7.7375 \times 10^{-34}$ & $1.018$ \\
		91 & $3.5644 \times 10^{-34}$ & $9.4985 \times 10^{-34}$ & $3.5010 \times 10^{-34}$ & $1.018$ \\
		92 & $1.6128 \times 10^{-34}$ & $4.2986 \times 10^{-34}$ & $1.5844 \times 10^{-34}$ & $1.018$ \\
		93 & $7.2985 \times 10^{-35}$ & $1.9457 \times 10^{-34}$ & $7.1717 \times 10^{-35}$ & $1.018$ \\
		94 & $3.3035 \times 10^{-35}$ & $8.8085 \times 10^{-35}$ & $3.2467 \times 10^{-35}$ & $1.017$ \\
		95 & $1.4955 \times 10^{-35}$ & $3.9884 \times 10^{-35}$ & $1.4701 \times 10^{-35}$ & $1.017$ \\
		96 & $6.7714 \times 10^{-36}$ & $1.8062 \times 10^{-35}$ & $6.6576 \times 10^{-36}$ & $1.017$ \\
		97 & $3.0665 \times 10^{-36}$ & $8.1812 \times 10^{-36}$ & $3.0155 \times 10^{-36}$ & $1.017$ \\
		98 & $1.3889 \times 10^{-36}$ & $3.7062 \times 10^{-36}$ & $1.3661 \times 10^{-36}$ & $1.017$ \\
		99 & $6.2919 \times 10^{-37}$ & $1.6792 \times 10^{-36}$ & $6.1894 \times 10^{-37}$ & $1.017$ \\
		100 & $2.8507 \times 10^{-37}$ & $7.6094 \times 10^{-37}$ & $2.8048 \times 10^{-37}$ & $1.016$ \\
		\hline
	\end{tabular}
\end{table}

\end{document}